\RequirePackage{fix-cm}

\documentclass[sn-mathphys]{sn-jnl}
\jyear{2021}%

\newcommand{\Rea}{\mathbb{R}}
\newcommand{\GENEO}{\mathcal{F}_{\mathrm{all}}}
\newcommand{\dmatch}{d_{\mathrm{match}}}
\newcommand{\LipL}{\mathrm{Lip}_L}

\theoremstyle{thmstyleone}%
\newtheorem{theorem}{Theorem}%
\newtheorem{proposition}[theorem]{Proposition}%
\newtheorem{lemma}{Lemma}%
\newtheorem{corollary}[theorem]{Corollary}%

\theoremstyle{thmstyletwo}%

\theoremstyle{thmstylethree}%
\newtheorem{definition}{Definition}%

\begin{document}
	
	\title[Noise reduction through GENEOs]{A probabilistic result on impulsive noise reduction in Topological Data Analysis through Group Equivariant Non-Expansive Operators}
	
	
	\author[1,2,3,4]{\fnm{Patrizio} \sur{Frosini}}\email{patrizio.frosini@unibo.it}
	
	\author[1]{\fnm{Ivan} \sur{Gridelli}}\email{ivangridelli@gmail.com}
	
	\author[1]{\fnm{Andrea} \sur{Pascucci}}\email{andrea.pascucci@unibo.it}
	
	\affil[1]{\orgdiv{Department of Mathematics}, \orgname{University of Bologna}, \orgaddress{\country{Italy}}}
	\affil[2]{\orgdiv{ALMA-AI}, \orgname{University of Bologna}, \orgaddress{\country{Italy}}}
    \affil[3]{\orgdiv{AM\textsuperscript{2}}, \orgname{University of Bologna}, \orgaddress{\country{Italy}}}
    \affil[4]{\orgdiv{ARCES}, \orgname{University of Bologna}, \orgaddress{\country{Italy}}}

	\abstract{In recent years, group equivariant non-expansive operators (GENEOs) have attracted attention in the fields of Topological Data Analysis and Machine Learning.
		In this paper we show how these operators can be of use also for the removal of impulsive noise and to increase the stability of TDA in the presence of noisy data.
		In particular, we prove that GENEOs can control the expected value of the perturbation of persistence diagrams caused by uniformly distributed impulsive noise, when data are represented by $L$-Lipschitz functions from $\Rea$ to $\Rea$.}
	
	
	
	\maketitle
	
	\section{Introduction}
	
	In the last thirty years Topological Data Analysis (TDA) developed as a useful mathematical theory for analyzing data, benefiting from the reduction of dimensionality guaranteed by topology \cite{EdHa08,BiDFFaal08,Ca2009}.
	One of the main tools in TDA is the concept of persistence diagram, which is a collection of points in the real plane, describing the homological changes of the sublevel sets of suitable continuous functions. These changes give important information about the data of interest, focusing on some of their most relevant properties. Persistence diagrams can be used in the presence of noise, since a well-known stability theorem states that these topological descriptors change in a controlled way when we know that the functions expressing the filtrations we are interested in change in a controlled way with respect to the sup-norm \cite{CSEdHa07}. Furthermore, $L^p$-stability of persistence diagrams with respect to the sup-norm have been proved in \cite{CoEdHaMi10}. Unfortunately, in many applications the sup-norm of the noise is not guaranteed to be small, and hence these results cannot be directly applied. In particular, these results cannot be directly used when data are affected by impulsive noise.
	
	Analogously, in the discrete setting of TDA the presence of outliers in cloud points can drastically affect the corresponding persistence diagrams. The problem of managing outliers has been studied by several authors by different techniques.
	In \cite{FaLeRiWaBaSi14} an approach based on confidence sets has been introduced.
	A method inspired from the $k$-nearest neighbors regression and the local median filtering has been used in \cite{BuChDeFaOuWa15},
	while the concept of bagplot has been applied in \cite{AdAg19}.
	In \cite{ViFuKuSrBh20} an approach based on reproducing kernels has been proposed.
	
	In our paper we start exploring a different probabilistic approach, in the topological setting. The main idea is that studying the properties of data is a process that should be primarily based on the analysis of the observers that are asked to examine the data, since we cannot ignore that different observers can differently judge the same data. This approach has been initially proposed in \cite{BeFrGiQu19} and requires both the definition of the space of Group Equivariant Non-Expansive Operators (GENEOs) and the development of geometrical techniques to move around in this space \cite{CoFrQu22}. In other words, in this model we should not wonder how we have to manage the data but rather how we have to manage the observers (i.e., GENEOs) analyzing the data.
	As an initial step in this direction, in this paper we show how GENEOs can be used to get stability of persistence diagrams of 1D-signals in the presence of impulsive noise.
	
	
	GENEOs have been studied in \cite{FrJa16} as a new tool in TDA, since they allow for an extension of the theory that is not invariant under the action of every homeomorphism of the considered domain. This is important in applications where the invariance group is not the group of all homeomorphisms, such as the ones concerning shape comparison. Interestingly, GENEOs are also deeply related to the foliation method used to define the matching distance in 2-dimensional persistent homology \cite{CeDFFeal13,CeEtFr19} and can be seen as a theoretical bridge between TDA and Machine Learning \cite{BeFrGiQu19}. Furthermore, these operators make available lower bounds for the natural pseudo-distance $d_G(\varphi_1,\varphi_2) := \inf\limits_{g\in G} \|\varphi_1-\varphi_2\circ g\|_\infty$, associated with a group $G$ of self-homeomorphisms of the domain of the signals $\varphi_1,\varphi_2$ \cite{FrJa16}.
	
	In our paper, we prove that GENEOs can control the expected value of the perturbation of persistence diagrams caused by uniformly distributed impulsive noise, when data are represented by $L$-Lipschitz functions from $\Rea$ to $\Rea$.
	In order to do that, we choose a ``mother'' bump function $\psi$, i.e. a continuous function that is non-negative, upper bounded by $1$ and compactly supported, and we assume that our noise is made up of finite bumps each obtained by translating, heightening and/or widening $\psi$. The function  $\hat\varphi=\varphi+R$ represents the corrupted data, where the noise is given  by the function $R=\sum\limits_{i=1}^k a_i\psi(b_i(x-c_i))$ for some $a_i,b_i,c_i\in\Rea$, with $b_i>0$, and some positive integer $k$.
	
	In this situation, trying to use a convolution to approximate our starting data is not effective, because even if it does contract the bumps, it does not cut them, and hence does not improve the sup-norm distance from the original data.
	A classical approach for the removal of impulsive noise is the one of using a median filter\cite{Va08}. Although this approach would be quite efficient in the discrete case, let us remark that in our setting we are considering continuous functions. The analogous of the median for the continuous case would be defined as the interval of those values $m$ such that $$\int\limits_{-\infty}^{m} f(x)dx=\int\limits_{m}^{+\infty} f(x)dx$$ where $f$ is a density of probability. However, this operator is not stable: a small alteration of the starting function could lead to a significant change of the median.
	
	The operators we consider in this paper are $F^\delta=\max (\varphi(x-\delta),\varphi(x+\delta))$ and $F_\varepsilon=\min (\varphi(x-\varepsilon),\varphi(x+\varepsilon))$. The main idea is that $F^\delta$ cuts the noise ``directed downwards'' and $F_\varepsilon$ the noise ``directed upwards'', and hence their composition should be able to eliminate all the bumps. These operators are GENEOs with respect to isometries of the real line. We prove that moving in the space of GENEOs by taking suitable values for $\varepsilon$ and $\delta$, we can get quite close to restoring the original function $\varphi$, depending on how the bumps are positioned. The closer the bumps are to being Dirac delta functions and the further they are from each other, the better our approximation can be.
	On the ground of this result, we finally get an estimate of the expected value $E(\|F^{\delta} \circ F_\varepsilon (\hat\varphi)-\varphi\|_\infty)$.
	
	The paper is structured as follows. In Section 2 the mathematical background is laid. The case we consider and our notation are explained in Section 3. In Section 4 we prove the results that are needed in order to demonstrate our main result. In Section 5 the main theorem giving us a probabilistic upper bound is formulated. In Section 6 some examples and experiments are presented in order to better illustrate the use of our results. A brief discussion concludes the paper.
	
	\section{Mathematical setting}
	In this section we will recall some basic concepts we will use in this paper.
	
	\subsection{Representing data as real functions}
	Let us consider a set $\varPhi$ of bounded functions from a set $X$ to $\Rea$, which will represent the data we wish to compare. We shall call $\varPhi$ the set of \emph{admissible measurements} on $X$. We endow $\varPhi$ with the topology induced by the sup-norm $\|\cdotp\|_{\infty}$ and the corresponding distance $D_\varPhi$.
	A pseudo-metric $D_X$ can be defined on $X$ by setting $D_X(x_1,x_2)=\sup\limits_{\varphi\in\varPhi}\lvert \varphi(x_1)-\varphi(x_2)\rvert$ for every $x_1,x_2\in X$.
	We recall that a pseudo-metric on a set $X$ is  a distance $d$ without the property $d(x,y)=0\implies x=y$.
	We will consider the topological space $(X,\tau_{D_X})$ where $\tau_{D_X}$ is the topology induced by $D_X$. A base for this topology is given by the open balls $\{B(x,r):=\{x'\in X:D_X(x,x')<r\}$ with $x\in X$, $r\in\Rea\}$.
	The choice of this topology makes every function in $\varPhi$ a continuous functions. As shown in \cite{BeFrGiQu19}, this fact enables us to use persistence diagrams in the study of $\varPhi$.
	
	\subsection{GENEOs as operators acting on data}
	We are interested in considering transformations of data. Let $\mathrm{Homeo}_\varPhi(X)$ be the set of $\varPhi$-preserving homeomorphisms from $X$ to $X$ with respect to the topology $\tau_{D_X}$, meaning that every $g$ in $\mathrm{Homeo}_\varPhi(X)$ is a homeomorphism of $X$ such as both $\varphi \circ g$ and $\varphi \circ g^{-1}$ belong to $\varPhi$ for every $\varphi$ in $\varPhi.$
	Let $G$ be a subgroup of $\mathrm{Homeo}_\varPhi(X)$. $G$ represents the set of transformations on data for which we will require equivariance to be respected.
	
	Under the previously stated assumptions we call the ordered pair $(\varPhi,G)$ a \emph{perception pair}.
	We can now introduce the concept of GENEO.
	
	\begin{definition}
		Let $(\varPhi,G)$ and $(\Psi,H)$ be perception pairs and assume that a homomorphism $T:G \to H$ is given.
		A function $F:\varPhi \to \Psi$ is called a \emph{Group Equivariant Non-Expansive Operator (GENEO) from $(\varPhi,G)$ to $(\Psi,H)$ with respect to $T$} if
		the following properties hold:
		\begin{enumerate}
			\item (Group Equivariance) $F(\varphi \circ g) =F(\varphi) \circ T(g)$ for every $\varphi \in \varPhi,g \in G$;
			\item (Non-Expansivity) $D_\Psi(F(\varphi_1),F(\varphi_2))\leq D_\varPhi(\varphi_1,\varphi_2)$ for every $\varphi_1,\varphi_2\in\varPhi$.
		\end{enumerate}
	\end{definition}
	
	
	Let us now consider the set $\GENEO$ of all GENEOs from $(\varPhi,G)$ to $(\Psi,H)$ with respect to $T:G\to H$.
	The space $\GENEO$ is endowed with the extended pseudo-metric $D_{\GENEO}$, defined by setting $D_{\GENEO}(F_1,F_2)=\sup_{\varphi\in\varPhi} D_\Psi(F_1(\varphi),F_2(\varphi))$ for every $F_1,F_2\in\GENEO$. The word \emph{extended} refers to the possibility that $D_{\GENEO}$ takes an infinite value.
	
	The following result can be proven \cite{BeFrGiQu19}:
	\begin{theorem}
		If $(\varPhi,D_\varPhi)$,$(\Psi,D_\Psi)$ are compact and convex then the metric space $(\GENEO,D_{\GENEO})$ is compact and convex.
	\end{theorem}
	
	If a non-empty set $\mathcal{F}\subseteq \GENEO$ is fixed, we can define the following pseudo-distance $D_{\mathcal{F},\varPhi}$ on $\varPhi$:
	\begin{definition}
		For any $\varphi_1,\varphi_2$ in $\varPhi$ we set
		$$D_{\mathcal{F},\varPhi}(\varphi_1,\varphi_2)=\sup_{F\in\mathcal{F}} \|F(\varphi_1)-F(\varphi_2)\|_\infty.$$
	\end{definition}
	This pseudo-distance allows us to compare data by taking into account how agents operate on data.
	Notice how, if $G$ becomes larger, the natural pseudo-distance $d_G$ becomes harder to compute but this new pseudo-distance $D_{\mathcal{F},\varPhi}$ becomes easier to evaluate.
	In order to find a lower bound for $D_{\mathcal{F},\varPhi}$ it is useful to introduce the notion of persistence diagram.
	
	\subsection{Persistence Diagrams}
	We will now recall some basic definitions and results in persistent homology. The interested reader can find more details in \cite{EdHa08}.
	
	Let us consider an ordered pair $(X,\varphi)$ where $X$ is a topological space and $\varphi:X\to\Rea$ is a continuous function. For any $t\in\Rea$ we can set $X_t:=\varphi^{-1}(\left]-\infty,t\right])$. If $u<v$ the inclusion $i_{u,v}:X_u\to X_v$ induces a homomorphism $i^k_{u,v}:H_k(X_u)\to H_k(X_v)$ between the $k$\textsuperscript{th} homology groups of $X_u$ and $X_v$. We can define the \emph{$k$\textsuperscript{th} persistent homology group}, with respect to $\varphi$ and computed at the point $(u,v)$, as $PH_k(u,v):=i^k_{u,v}(H_k(X_u))$.
	Moreover, we can define the \emph{$k$\textsuperscript{th} persistent Betti numbers function} $r_k(u,v)$ as the rank of $PH_k(u,v)$.
	
	The $k$\textsuperscript{th} persistent Betti numbers function can be represented by the $k$\textsuperscript{th} persistence diagram. This diagram is defined as the multi-set of all the ordered pairs $(u_j, v_j),$ where $u_j$ and $v_j$ are the times of birth and death of the $j$\textsuperscript{th} $k$-dimensional hole in $X$, respectively. We call \emph{time of birth} of a hole the first time at which the homology class appears, and \emph{time of death} the first time at which the homology class merges with an older one. When a hole never dies, we set its time of death equal to $\infty$. We also add to this set all points of the form $(w,w)$ for $w\in\Rea$.
	
	In Figure~\ref{filtration} the filtration of the set $X:=\left[0,\frac{3}{4}\pi\right]$ given by the function $\varphi(x)=2\sin x$ is illustrated. In this example, the topology on $X$ is the one defined by the space $\varPhi$ of admissible functions given by all functions $a\sin(x+b)$ with $a,b\in \Rea$. The reader can easily check that this topology coincides with the Euclidean topology. The persistence diagram in degree $k=0$ of the function $\varphi$ is displayed in Figure~\ref{pd}.
	
	
	\begin{figure}[t]
		\begin{center}
			\begin{tabular}{r c}
			\includegraphics[width=4.8cm]{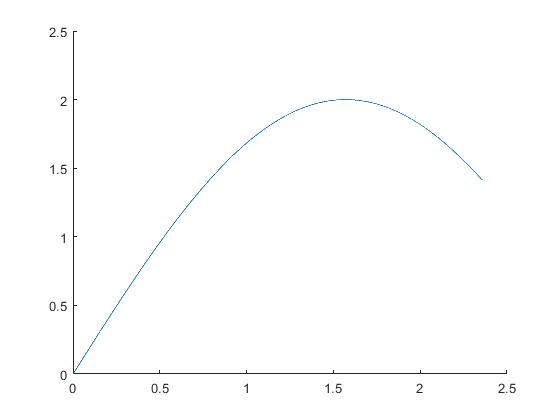} & \includegraphics[width=4.8cm]{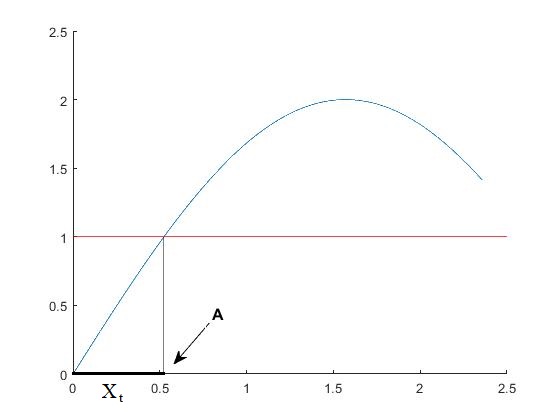} \\
			\includegraphics[width=4.8cm]{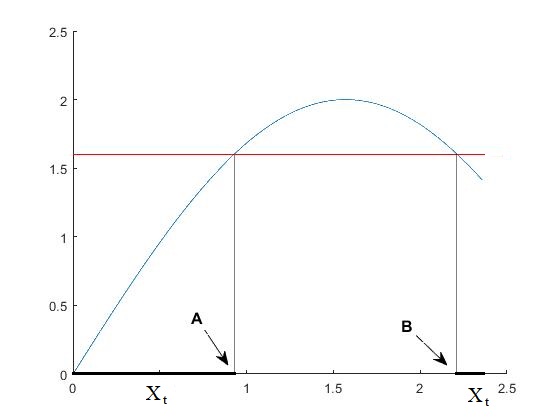} & \includegraphics[width=4.8cm]{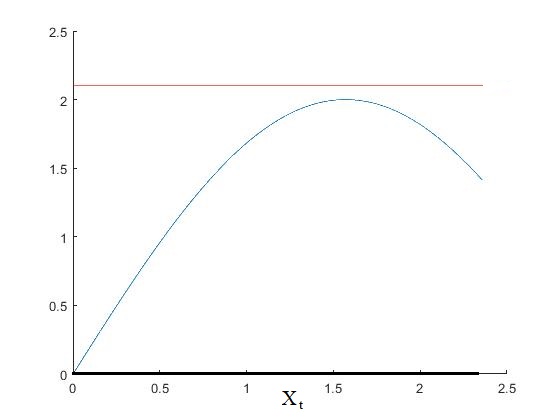}  \\
			\end{tabular}
			\caption{How the sublevel sets change with respect to the filtration induced by $\varphi$.}
			\label{filtration}
		\end{center}
	\end{figure}
	
	\begin{figure}[t]
		\begin{center}
			\includegraphics[width=4.8cm]{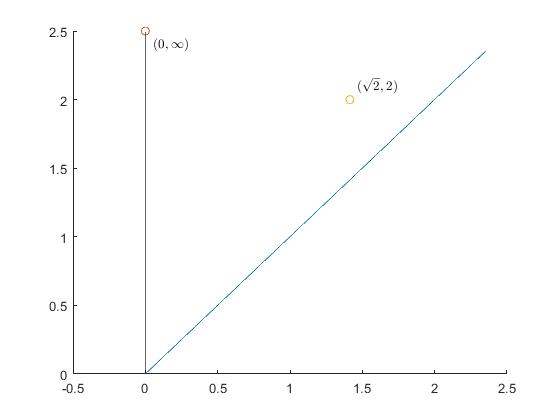}
		\end{center}
		\caption{Persistence diagram of the function $2\sin x$ on $\left[0,\frac{3}{4}\pi\right]$.
			The point $(0,\infty)$ describes the existence of a connected component that is born at zero and never dies.
			The point $(\sqrt{2},2)$ claims that there is a connected component born at $\sqrt{2}$ that dies (merges with the other one) at 2.
			The trivial points on the diagonal $u=v$ are not displayed.}
		\label{pd}
	\end{figure}
	
	\subsection{Comparing Persistence Diagrams}
	Persistence diagrams can be efficiently compared by means of a suitable metric $\dmatch$.
	In order to define it, we first define the pseudo-distance $$\delta((x,y),(x',y'))= \min\{\max\{\lvert x-x'\rvert,\lvert y-y'\rvert\},\max\{\lvert x-y\rvert/2,\lvert x'-y'\rvert/2\}\}$$
	for all $(x,y),(x',y')\in\{(x,y)\in\Rea$ with $x\leq y\}\cup\{(x,\infty)$ with $x\in\Rea\}$ by agreeing that $\infty-y=\infty, y-\infty=-\infty$ for $y\not=\infty, \infty-\infty=0, \infty/2=\infty, \lvert\pm\infty\rvert =\infty,  \min\{\infty,c\} = c,\max\{\infty, c\} = \infty$.
	
	If two persistence diagrams $D,D'$ are given, we can set
	$$\dmatch(D,D')=\inf\limits_{\sigma\in\Sigma} \sup\limits_{P\in D} \delta(P,\sigma(P))$$ where $\Sigma$ represents the set of all bijections between the multisets $D,D'$.
	
	For every degree $k$ we can now define a new pseudo-metric:
	$$D_{\mathcal{F},\varPhi}^{\mathrm{match}}(\varphi_1,\varphi_2)=\sup_{F\in\mathcal{F}} \dmatch(D_{F(\varphi_1)},D_{F(\varphi_2)})$$
	where $D_{F(\varphi_1)},D_{F(\varphi_2)}$ are the persistence diagrams at degree $k$ of the functions $F(\varphi_1),F(\varphi_2)$, respectively.
	
	In this paper we will limit ourselves to considering data represented as functions from $\Rea$ to $\Rea$, and we recall that for this kind of data persistence diagrams are non-trivial only in degree $k=0$ (i.e., when persistent homology is used to count connected components). For this reason, in the following we will always assume $k=0$.

	
	\section{Our model}
	\begin{figure}[t]
		\begin{center}
			\begin{tabular}{r c}
			\includegraphics[width=4.8cm]{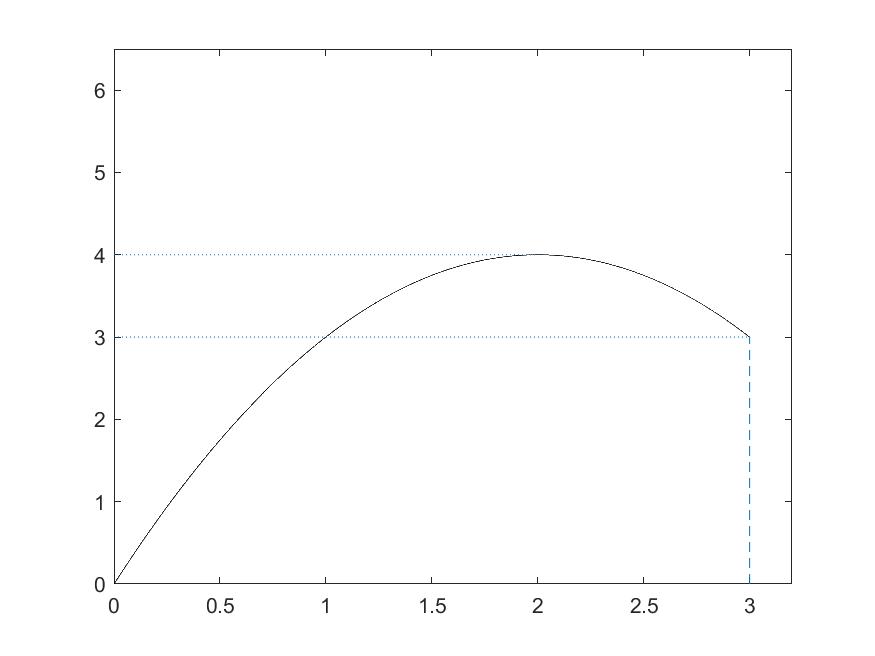} & \includegraphics[width=4.8cm]{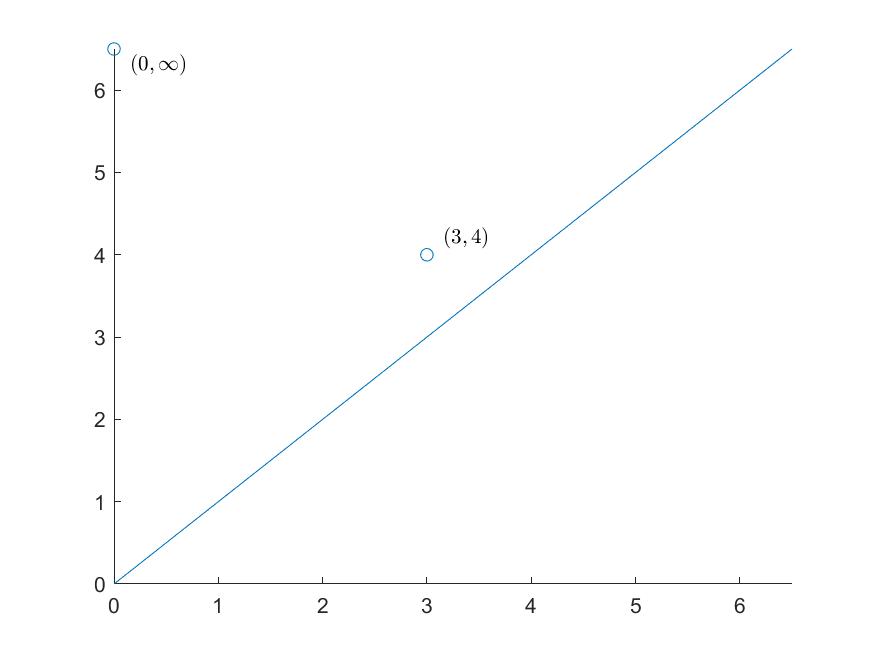} \\
			\includegraphics[width=4.8cm]{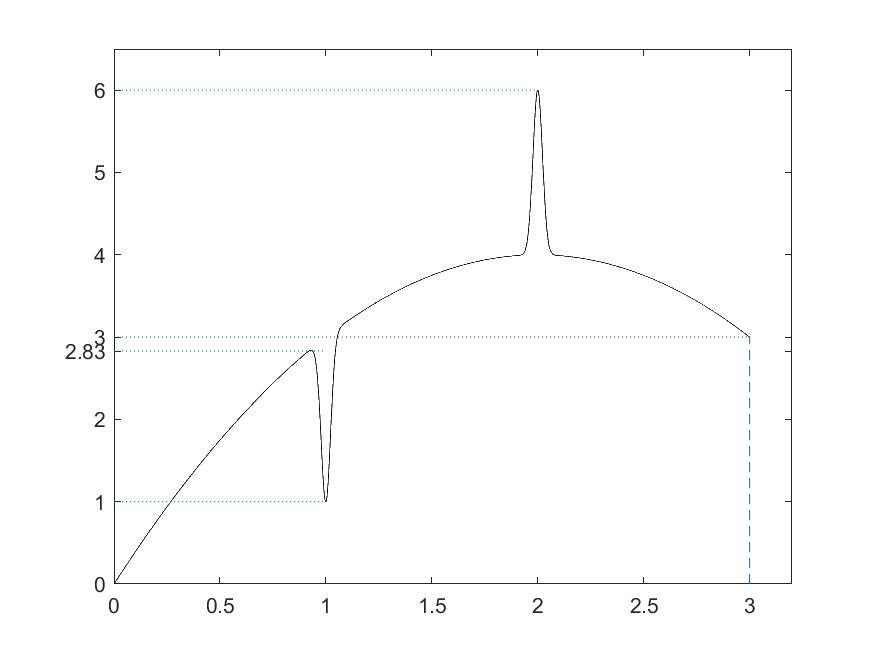} & \includegraphics[width=4.8cm]{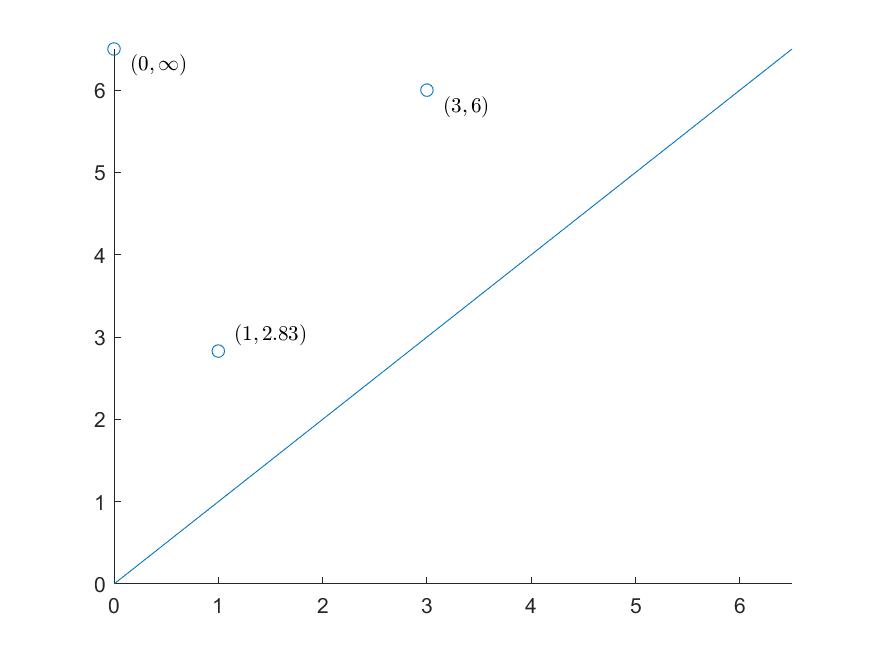}  \\
			\end{tabular}
			\caption{How drastically impulsive noise can influence persistence diagrams.}
		\end{center}
	\end{figure}
	In this paper we will be mainly interested in the set $\LipL$ of all $L$-Lipschitz functions from $\Rea$ to $\Rea$, for some fixed constant $L\in\Rea$, and in the set $C^0(\Rea)$ of all functions from $\Rea$ to $\Rea$ that are continuous with respect to the Euclidean topology.
	We will set $X=\Rea$ and consider the perception pairs $(\LipL,G)$, $(C^0(\Rea),G)$, where $G$ is the group of the Euclidean isometries of $\Rea$.
	
	We will assume our noise to be made up of a finite number of copies of a ``mother'' nonnegative continuous bump function $\psi:\Rea\to\Rea$, such that $\mathrm{supp}(\psi)\subseteq\left]-\sigma,\sigma\right[$ for some $\sigma>0$ and $\|\psi\|_\infty\leq 1$. We recall that the support of a function is the closure of the set of points where $f$ is non-vanishing.
	After fixing two positive real numbers $\eta,\beta$, the noise we will be adding is a function $R$ belonging to the space $\mathcal{R}_{\eta,\beta}$ that contains the null function and all functions of the form $\sum\limits_{i=1}^k a_i\psi(b_i(x-c_i))$, where $k$ is a positive integer and $a_i,b_i,c_i$ are real numbers such that $\lvert c_i-c_j\rvert\ge\eta$ for $i\not=j$ and $b_i\ge\beta$ for every index $i$. For any $R\in\mathcal{R}_{\eta,\beta}$ we define $S(R):=\bigcup\limits_{i=1}^k \left]c_i-\frac{\sigma}{\beta},c_i+\frac{\sigma}{\beta}\right[$ and remark that if $x\not\in S(R)$ then $R(x)=0$.
	
	Our purpose will be the one of recovering $\varphi\in\LipL$ as well as possible from the function $\hat\varphi=\varphi+R$. An example of such situation is depicted in Figure~\ref{corruptedphi}.
	
	\begin{figure}[t]
		\begin{tabular}{c c c} \includegraphics[width=4cm]{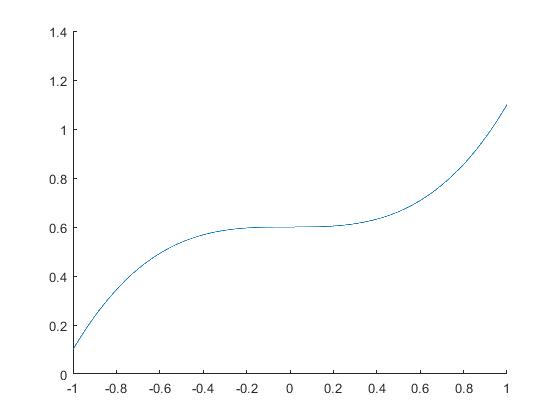} &  \includegraphics[width=4cm]{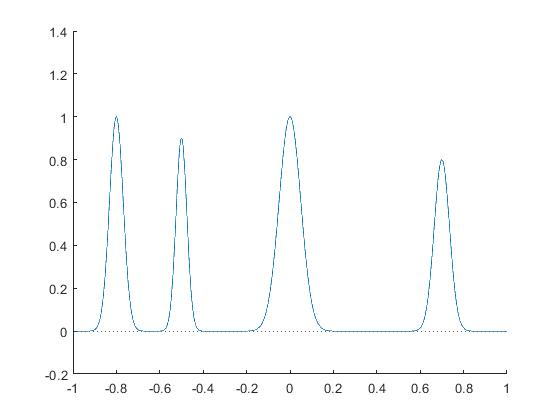} & \includegraphics[width=4cm]{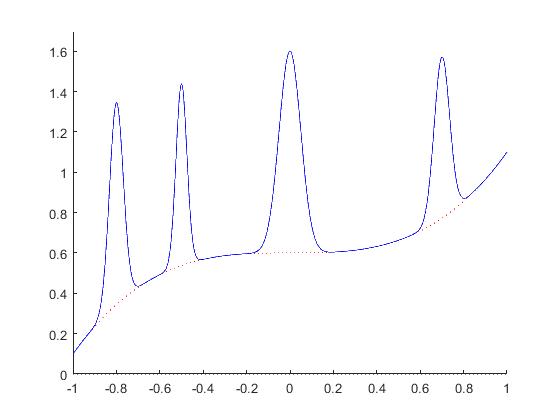}
		\end{tabular}
		\caption{In the figure on the left we have our original function $\varphi$, in the middle a noise function R and in the right figure the corrupted function $\hat\varphi:=\varphi+R$.}\label{corruptedphi}
	\end{figure}
	
	The following result will be of use.
	\begin{proposition}\label{propcomp}
		Let $F_1,F_2$ be GENEOs from $(C^0(\Rea),G)$ to $(C^0(\Rea),G)$ with respect to the trivial homomorphism $T=\mathrm{id}:G\to G$.
		Then $F_1\circ F_2$ is a GENEO from $(C^0(\Rea),G)$ to $(C^0(\Rea),G)$ with respect to $T$.
	\end{proposition}
	
	\begin{proof}
		For every $\varphi \in C^0(\Rea),g \in G$ we have that
		\begin{align*}
			F_1 \circ F_2(\varphi \circ g)&=F_1(F_2(\varphi \circ g)) \\
			&=F_1(F_2(\varphi) \circ g) \\
			&=F_1(F_2(\varphi)) \circ g \\
			&=(F_1 \circ F_2)(\varphi) \circ g.
		\end{align*}
		Therefore, $F_1\circ F_2$ is $G$-equivariant.
		Moreover, for any $\varphi_1,\varphi_2 \in C^0(\Rea)$
		\begin{align*}
			D_{C^0(\Rea)}(F_1 \circ F_2(\varphi_1)),F_1\circ F_2(\varphi_2))&=D_{C^0(\Rea)}(F_1(F_2(\varphi_1)),F_1(F_2(\varphi_2))) \\
			&\le D_{C^0(\Rea)}(F_2(\varphi_1)),F_2(\varphi_2)) \\
			&\le D_{C^0(\Rea)}(\varphi_1,\varphi_2).
		\end{align*}
		It follows that $F_1\circ F_2$ is non-expansive.
	\end{proof}
	
	\section{Cutting off the noise by GENEOs}
	
	We start by introducing two families of GENEOs from $(C^0(\Rea),G)$ to $(C^0(\Rea),G)$ with respect to the identical homomorphism.
	
	\begin{definition}
		Let $\varphi \in \LipL$ and $\varepsilon>0$. For all $x\in\Rea$ we define:
		\begin{enumerate}
			\item $F^\varepsilon(\varphi)(x)=\max (\varphi(x-\varepsilon),\varphi(x+\varepsilon))$;
			\item $F_\varepsilon(\varphi)(x)=\min (\varphi(x-\varepsilon),\varphi(x+\varepsilon))$.
		\end{enumerate}
	\end{definition}
	
	
	\begin{proposition}
		The maps $F^\varepsilon$ and $F_\varepsilon$ are GENEOs from $(C^0(\Rea),G)$ to $(C^0(\Rea),G)$ with respect to the identical homomorphism.
	\end{proposition}
	\begin{proof}
		
		We start by proving that $F^\varepsilon$ is $G$-equivariant.
		
		
		Let $g$ be the translation $x \mapsto x+k$ with $k \in \Rea$, then
		\begin{align*}
			F^\varepsilon(\varphi \circ g)
			&=\max\{\varphi((x+k)-\varepsilon),\varphi((x+k)+\varepsilon)\} \\
			&=\max\{\varphi((x-\varepsilon)+k),\varphi((x+\varepsilon)+k)\} \\
			&=F^\varepsilon(\varphi) \circ g.
		\end{align*}
		Let $g$ be the symmetry $x \mapsto -x$, then
		\begin{align*}
			F^\varepsilon(\varphi \circ g)
			&=\max\{\varphi((-x)-\varepsilon),\varphi((-x)+\varepsilon)\} \\
			&=\max\{\varphi(-(x+\varepsilon)),\varphi(-(x-\varepsilon))\} \\
			&=\max\{\varphi(x+\varepsilon),\varphi(x-\varepsilon)\} \circ g \\
			&=F^\varepsilon(\varphi) \circ g.
		\end{align*}
		Since every isometry in $G$ can be written as the composition of a symmetry and a translation, our statement follows.
		
		Furthermore, for any $x\in\Rea$
		\begin{align*}
			&\lvert F^\varepsilon(\varphi_1)(x)-F^\varepsilon(\varphi_2)(x)\rvert= \\
			&=\lvert\max\{\varphi_1(x-\varepsilon),\varphi_1(x+\varepsilon)\}-\max\{\varphi_2(x-\varepsilon),\varphi_2(x+\varepsilon)\}\rvert \\
			&\le\max\{\lvert\varphi_1(x-\varepsilon)-\varphi_2(x-\varepsilon)\rvert,\lvert\varphi_1(x+\varepsilon)-\varphi_2(x+\varepsilon)\rvert\} \\
			&\le\max\{\|\varphi_1-\varphi_2\|_\infty,\|\varphi_1-\varphi_2\|_\infty\} \\
			&=\|\varphi_1-\varphi_2\|_\infty.
		\end{align*}
		It follows that $\|F^\varepsilon(\varphi_1)-F^\varepsilon(\varphi_2)\|_\infty\le\|\varphi_1-\varphi_2\|_\infty$, and hence
		$F^\varepsilon$ is non-expansive.
		
		We observe that $\min\{a,b\}=-\max\{-a,-b\}$, and hence
		$F_\varepsilon(\varphi)=-F^\varepsilon(-\varphi)$ for every $\varphi\in C^0(\Rea)$. Moreover, the map taking $\varphi$ to $-\varphi$ is a GENEO from $(C^0(\Rea),G)$ to $(C^0(\Rea),G)$ with respect to the identical homomorphism.
		Therefore, Proposition~\ref{propcomp}
		guarantees that $G$-equivariance and non-expansivity also hold for $F_\varepsilon$.
	\end{proof}
	
	Since Proposition~\ref{propcomp} shows that the composition of GENEOs is still a GENEO, the operator $F^{\delta} \circ F_\varepsilon$ is a GENEO from $(C^0(\Rea),G)$ to $(C^0(\Rea),G)$ with respect to the identical homomorphism.
	
	We want now to prove that if a function $\hat\varphi$ is obtained by adding impulsive noise to a function $\varphi$ then the value of
	$$\|F^\delta \circ F_\varepsilon (\hat\varphi)-\varphi\|_\infty$$ is bounded, and possibly small, provided that $\delta$ and $\varepsilon$ are suitably chosen. The main idea is that the operator $F_\varepsilon$ cuts the noise ``directed upwards'' and $F^\delta$ cuts the noise ``directed downwards''.
	
	
	In order to proceed, we need two lemmas.
	\begin{lemma}
		\label{lemma1}
		Let $R\in C^0(\Rea,\Rea)$, $\varphi\in \LipL$ for some $L\in\Rea$, and set $\hat\varphi:=\varphi+R$.
		Then for any $\varepsilon>0$ and $\delta>0$
		\begin{description}
			\item[$\mathrm{i)}$] $-L\varepsilon+F_\varepsilon(R)\leq F_\varepsilon(\hat\varphi)-\varphi \leq L\varepsilon+F_\varepsilon(R)$;
			\item[$\mathrm{ii)}$] $-L\delta+F^\delta(R)\leq F^\delta(\hat\varphi)-\varphi \leq L\delta+F^\delta(R)$;
			\item[$\mathrm{iii)}$] $-L(\delta+\varepsilon)+F^\delta\circ F_\varepsilon(R)\leq F^\delta\circ F_\varepsilon(\hat\varphi)-\varphi \leq L(\delta+\varepsilon)+F^\delta\circ F_\varepsilon(R)$.
		\end{description}
	\end{lemma}
	\begin{proof}
		Since $\varphi$ is Lipschitz of constant $L$ we have that
		for any value $x\in\Rea$
		$\lvert\varphi(x-\varepsilon)-\varphi(x)\rvert\leq L\varepsilon$ and $\lvert\varphi(x+\varepsilon)-\varphi(x)\rvert\leq L\varepsilon$. Therefore,
		\begin{align*}
			F_\varepsilon(\hat\varphi)(x)&=F_\varepsilon(\varphi+R)(x) \\
			&=\min\{\varphi(x-\varepsilon)+R(x-\varepsilon),\varphi(x+\varepsilon)+R(x+\varepsilon)\} \\
			&\leq\min\{\varphi(x)+L\varepsilon+R(x-\varepsilon),\varphi(x)+L\varepsilon+R(x+\varepsilon)\} \\
			&=\varphi(x)+L\varepsilon+\min\{R(x-\varepsilon),R(x+\varepsilon)\} \\
			&=\varphi(x)+L\varepsilon+F_\varepsilon(R)(x).
		\end{align*}
		Analogously, $F_\varepsilon(\hat\varphi)(x)\geq\varphi(x)- L\varepsilon+F_\varepsilon(R)(x)$.
		The same steps applied to $F^\delta$ yield the second statement of the lemma.
		As for the last claim we can see that:
		\begin{align*}
			&F^\delta\circ F_\varepsilon(\hat\varphi)(x) \\
			&=\max\{F_\varepsilon(\hat\varphi)(x-\delta),F_\varepsilon(\hat\varphi)(x+\delta)\} \\
			&\leq\max\{\varphi(x-\delta)+L\varepsilon+F_\varepsilon(R)(x-\delta),\varphi(x+\delta)+L\varepsilon+F_\varepsilon(R)(x+\delta)\} \\
			&\leq\max\{\varphi(x)+L\delta+L\varepsilon+F_\varepsilon(R)(x-\delta),\varphi(x)+L\delta+L\varepsilon+F_\varepsilon(R)(x+\delta)\} \\
			&=\varphi(x)+L\delta+L\varepsilon+\max\{F_\varepsilon(R)(x-\delta),F_\varepsilon(R)(x+\delta)\} \\
			&=\varphi(x)+L\delta+L\varepsilon+F^\delta\circ F_\varepsilon(R)(x).
		\end{align*}
		Analogously, we can prove the lower bound.
	\end{proof}
	Henceforth we will assume that any summation on an empty set of indexes is the null function.
	\begin{lemma}
		\label{lemma2}
		Let $R\in\mathcal{R}_{\eta,\beta}$ and $\lambda\ge\frac{\sigma}{\beta}$.
		If $\lambda\le\rho\le\frac{\eta}{2}-\lambda$ then
		\begin{description}
			\item[$\mathrm{a)}$] $F_\rho(R)(x)=\sum\limits_{a_i<0} \left[a_i\psi(b_i(x-\rho-c_i))+a_i\psi(b_i(x+\rho-c_i))\right]\leq0$
			\item[$\mathrm{b)}$] $F^\rho(R)(x)=\sum\limits_{a_i>0} \left[a_i\psi(b_i(x-\rho-c_i))+a_i\psi(b_i(x+\rho-c_i))\right]\geq0$
		\end{description}
		for all $x\in\Rea$. Moreover $F_\rho(R),F^\rho(R)(x)\in\mathcal{R}_{2\lambda,\beta}$.
	\end{lemma}
	\begin{proof}
		We will suppose without loss of generality that $c_i<c_{i+1}$ for all $i=1,\dots,k-1$ and $a_i\not=0$ for all indexes $i$.
		
		We want to show that at least one of $x-\rho$ and $x+\rho$ must always be outside of $\bigcup\limits_{i=1}^k \left]c_i-\lambda,c_i+\lambda\right[$.
		If by contradiction both $x-\rho$ and $x+\rho$ were in the same $\left]c_i-\lambda,c_i+\lambda\right[$ then we would have $\rho<\lambda$, against our hypotheses.
		Moreover $x-\rho$ and $x+\rho$ cannot belong to different intervals $\left]c_i-\lambda,c_i+\lambda\right[$ and $\left]c_j-\lambda,c_j+\lambda\right[$ for $i<j$:
		if by contradiction $x-\rho\in\left]c_i-\lambda,c_i+\lambda\right[$
		and $x+\rho\in\left]c_j-\lambda,c_j+\lambda\right[$ for some $i<j$ then $2\rho=x+\rho-(x-\rho)>
		(c_j-\lambda)-(c_i+\lambda)\geq
		c_j-c_i-2\lambda\geq
		\eta-2\lambda$, and hence we would have $\rho>\frac{\eta}{2}-\lambda$, against our hypotheses.
		
		Since $\left]c_i-\frac{\sigma}{\beta},c_i+\frac{\sigma}{\beta}\right[\subseteq\left]c_i-\lambda,c_i+\lambda\right[$, it follows that at least one among the values $R(x-\rho)$ and $R(x+\rho)$ must always be zero. Let us now set $I^-_i:=\left](c_i-\rho)-\lambda,(c_i-\rho)+\lambda\right[$ and $I^+_i:=\left](c_i+\rho)-\lambda,(c_i+\rho)+\lambda\right[$. These two intervals must be disjoint, since $(c_i+\rho-\lambda)-(c_i-\rho+\lambda)=2\rho-2\lambda\ge0$.
		
		Let us now consider $\{I^-_1,I^+_1,\ldots,I^-_k,I^+_k\}$. We will now prove that any two distinct elements from this set must be disjoint. Since we have just proven that $I^-_i\cap I^+_i=\emptyset$ for $i=1,\dots,k$, the following holds in the case $i<j$:
		\begin{enumerate}
			\item $I^+_i\cap I^+_j=\emptyset$, since $c_j+\rho-\lambda-(c_i+\rho+\lambda)\geq c_j-c_i-2\lambda\ge\eta-2\lambda\ge0$;
			\item $I^-_i\cap I^-_j=\emptyset$, since $c_j-\rho-\lambda-(c_i-\rho+\lambda)\geq c_j-c_i-2\lambda\ge\eta-2\lambda\ge0$;
			\item $I^+_i\cap I^-_j=\emptyset$, since $c_j-\rho-\lambda-(c_i+\rho+\lambda)\geq c_j-c_i-2\rho-2\lambda\ge\eta-2\left(\frac{\eta}{2}-\lambda\right)-2\lambda=0$;
			\item $I^+_j\cap I^-_i$, since $c_j+\rho-\lambda-(c_i-\rho+\lambda)\geq c_j-c_i+2\rho-2\lambda\ge\eta+2\lambda-2\lambda\ge0$.
		\end{enumerate}
		Let us now fix $k\in\{1,\dots,k\}$.
		
		Since $\mathrm{supp}(a_k\psi(b_k(x-c_k)))\subseteq\left]c_k-\lambda,c_k+\lambda\right[$, then $\mathrm{supp}(a_k\psi(b_k(x-\rho-c_k)))\subseteq I^-_k$ and $\mathrm{supp}(a_k\psi(b_k(x+\rho-c_k)))\subseteq I^+_k$.
		
		Hence, for any $x\in\Rea$ we have that
		\begin{align*}
			F_\rho(R)(x)&=\min\{R(x-\rho),R(x+\rho)\} \\
			&=
			\begin{cases}
				0 &\mbox{if } x\notin\bigcup\limits_{i=1}^k (I_i^-\cup I_i^+) \\
				a_j\psi(b_j(x+\rho-c_j)) &\mbox{if } x\in I^-_j \text{ and } a_j<0 \\
				a_j\psi(b_j(x-\rho-c_j)) &\mbox{if } x\in I^+_j \text{ and } a_j<0.
			\end{cases}
		\end{align*}
		This means that $$F_\rho(R)(x)=\sum\limits_{a_i<0} \left[a_i\psi(b_i(x-\rho-c_i))+a_i\psi(b_i(x+\rho-c_i))\right]\leq0.$$
		Now, given $c_i+\rho$ and $c_j-\rho$ centers of bumps of $F_\rho(R)$, we have that $$\lvert(c_i+\rho)-(c_j-\rho)\rvert\geq\min\{\eta-2\rho,2\rho\}\geq2\lambda.$$ It follows that $F_\rho(R)\in\mathcal{R}_{2\lambda,\beta}$.
		
		By noticing that $F^\rho(R)=-F_\rho(-R)$ we get the second part of the thesis.
	\end{proof}
	
	Let us remark that, in particular, if $\lambda=2\frac{\sigma}{\beta}$ then $F_\rho(R),F^\rho(R)\in\mathcal{R}_{4\frac{\sigma}{\beta},\beta}$.
	We observe that in order to exist a $\rho$ that satisfies the hypotheses of Lemma~\ref{lemma2}, the inequality $\eta\ge4\lambda$ must hold.
	
	We are now actually ready to prove a key result in our paper.
	
	\begin{theorem}
		\label{Upper Bound}
		Given $\theta,\beta,L\in\Rea$, let ${\bar R}\in\mathcal{R}_{\theta,\beta}$, $\varphi\in \LipL$ and set $\hat\varphi:=\varphi+{\bar R}$.
		If $2\frac{\sigma}{\beta}\le\varepsilon\le\frac{\theta}{2}-2\frac{\sigma}{\beta}$ then for any $\frac{\sigma}{\beta}\le\delta\le\frac{1}{2}\min\{\theta-2\varepsilon,2\varepsilon\}-\frac{\sigma}{\beta}$ the following inequality holds:
		$$\|F^{\delta} \circ F_\varepsilon (\hat\varphi)-\varphi\|_\infty\leq L(\varepsilon+\delta).$$
	\end{theorem}
	
	\begin{proof}
		Let $x\in\Rea$.
		Firstly, let us notice that the condition $2\frac{\sigma}{\beta}\le \varepsilon\le \frac{\theta}{2}-2\frac{\sigma}{\beta}$ implies that $\frac{\sigma}{\beta}\le \frac{1}{2}\min\{\theta-2\varepsilon,2\varepsilon\} - \frac{\sigma}{\beta}$.
		From Lemma~\ref{lemma1} we know that $$-L(\delta+\varepsilon)+F^\delta\circ F_\varepsilon({\bar R})(x)\leq F^\delta\circ F_\varepsilon(\hat\varphi)(x)-\varphi(x) \leq L(\delta+\varepsilon)+F^\delta\circ F_\varepsilon({\bar R})(x).$$
		Let us now prove that $F^\delta\circ F_\varepsilon({\bar R})(x)=0$.
		
		By applying Lemma~\ref{lemma2} with $\eta:=\theta$, $\lambda:=2\frac{\sigma}{\beta}$, $\rho:=\varepsilon$ and $R:={\bar R}$ we get
		$$F_\varepsilon({\bar R})(x)=\sum\limits_{a_i<0} \left[a_i\psi(b_i(x-\varepsilon-c_i))+a_i\psi(b_i(x+\varepsilon-c_i))\right]\le0.$$ Let us remark that $F_\varepsilon({\bar R})\in\mathcal{R}_{4\frac{\sigma}{\beta},\beta}$. Moreover,
		\begin{enumerate}
			\item $F^\delta\circ F_\varepsilon({\bar R})(x)=\max\{F_\varepsilon({\bar R})(x-\delta),F_\varepsilon({\bar R})(x+\delta)\}\le0$ since both terms are negative;
			\item $F^\delta \circ F_\varepsilon({\bar R})(x)=F^\delta(F_\varepsilon({\bar R}))(x)\ge0$.
		\end{enumerate}
		These inequalities follow from Lemma~\ref{lemma2}, by setting $\eta:=4\frac{\sigma}{\beta}$, $\lambda:=\frac{\sigma}{\beta}$, $\rho:=\delta$ and $R:=F_\varepsilon({\bar R})$.
		
		Therefore, we have proved that $\lvert F^\delta\circ F_\varepsilon(\hat\varphi)(x)-\varphi(x)\rvert\leq L(\delta+\varepsilon$) for any $x\in\Rea$. It follows that $\|F^{\delta} \circ F_\varepsilon (\hat\varphi)-\varphi\|_\infty\leq L(\varepsilon+\delta)$.
	\end{proof}
	
	Let us remark that Theorem~\ref{Upper Bound} works under the (only) implicit assumption that $\theta\ge8\frac{\sigma}{\beta}$. In our setting this should not be restrictive since it means that the noise added is made up of scattered, thin bumps, without any reference to the height of the bumps: this is what we expect when considering additive impulsive noise.
	
	\begin{corollary} \label{coroll2}
		Given $\theta,\beta,L\in\Rea$, let $R\in\mathcal{R}_{\theta,\beta}$, $\varphi\in \LipL$, and set $\hat\varphi:=\varphi+R$. If $\theta\ge8\frac{\sigma}{\beta}$,
		then $$\left\|F^{\frac{\sigma}{\beta}} \circ F_{2\frac{\sigma}{\beta}} \left(\hat\varphi\right)-\varphi\right\|_\infty\leq 3L\frac{\sigma}{\beta}.$$
	\end{corollary}
	
	\begin{proof}
		The claim follows from Theorem~\ref{Upper Bound} by taking $\varepsilon:=2\frac{\sigma}{\beta}$ and $\delta:=\frac{\sigma}{\beta}.$
	\end{proof}
	
	Corollary~\ref{coroll2} and the well known stability of persistence diagrams with respect to the max-norm \cite{CSEdHa07} immediately imply the following result, which is of interest in TDA (the symbol $\dmatch$ denotes the usual bottleneck distance between persistence diagrams).
	
	\begin{corollary}
		\label{coroll3}
		Given $\theta,\beta,L\in\Rea$, let $R\in\mathcal{R}_{\theta,\beta}$, $\varphi\in \LipL$, and set $\hat\varphi:=\varphi+R$.
		If $\theta\ge8\frac{\sigma}{\beta}$, and $D$ and $D'$ are the persistence diagrams in degree $0$ of the filtering functions $\varphi$ and $F^{\frac{\sigma}{\beta}}\circ F_{2\frac{\sigma}{\beta}}(\hat\varphi)$, respectively, then
		$$\dmatch(D,D')\le 3L\frac{\sigma}{\beta}.$$
	\end{corollary}

	\section{Our main result}
	
	We are now ready to prove our main result. We start by stating a lemma concerning the probability
	$p$ that any two distinct points in a randomly chosen set of cardinality $k$ in an interval of
	length $\ell$ have a distance greater than $\eta$. A proof of the following lemma
	is provided in \cite{2001026}: for the reader's convenience, we report it here.
	\begin{lemma}\label{lemma3}
		Let $X_{1},\dots,X_{k}$, with $k\ge 2$, be independent random variables, uniformly distributed on
		the interval $[0,\ell]$, for some $\ell>0$. Let
		$$M:=\min_{1\le i,j\le k\atop i\neq j}\lvert X_{i}-X_{j}\rvert$$
		be the minimal distance between two distinct random variables. Then we have
		$$P(M>\eta)=
		\begin{cases}
			1 & \text{if } \eta\le 0,\\
			\left(1-\frac{(k-1)\eta}{\ell}\right)^{k} & \text{if } 0<\eta<\frac{\ell}{k-1}, \\
			0 & \text{if } \eta\ge\frac{\ell}{k-1}.
		\end{cases}
		$$
	\end{lemma}
	\begin{proof}
		It suffices to consider the case $0<\eta<\frac{\ell}{k-1}$. By symmetry, we have
		\begin{align}\nonumber
			P(M>\eta)&=k!\,P((M>\eta)\cap(X_{1}<X_{2}<\cdots<X_{k}))
			\intertext{(since $X_{1},\dots,X_{k}$ are uniformly distributed)}\label{a1}
			&=k!\,\frac{\text{Leb}(S)}{\ell^{k}}
		\end{align}
		where $S=\{x\in[0,\ell]^{k}\mid x_{1}<x_{2}-\eta<x_{3}-2\eta<\cdots<x_{k}-(k-1)\eta\}$ and
		$\text{Leb}$ denotes the Lebesgue measure. Setting $y_{i}=x_{i}-(i-1)\eta$ for $i=1,\dots,k$, we
		have that $\text{Leb}(S)=\text{Leb}(S')$ where
		$$S'=\{y\in[0,\ell-(k-1)\eta]^{k}\mid y_{1}<y_{2}<\cdots<y_{k}\}.$$
		On the other hand, again by symmetry, we have
		$$\text{Leb}(S')=\frac{\text{Leb}([0,\ell-(k-1)\eta]^{k})}{k!}=\frac{(\ell-(k-1)\eta)^{k}}{k!}$$
		and plugging this last identity into \eqref{a1} we get the thesis.
	\end{proof}

	
	We can now prove the following result, concerning the expected value of the error $\left\|
		F^{\frac{\sigma}{\beta}}\circ F_{2\frac{\sigma}{\beta}}(\hat\varphi)-\varphi
		\right\|_\infty$.
	
	\begin{theorem}
		\label{main_result}
Let us choose a function $\varphi\in\LipL$, a non-negative continuous function $\psi:\Rea\to\Rea$ with $\|\psi\|_\infty\leq 1$ and $\mathrm{supp}(\psi)\subseteq\left]-\sigma,\sigma\right[$ for some $\sigma>0$, two positive numbers $\beta$ and $\ell$, and an integer $k\ge 2$.
For $i=1,\ldots,k$, let us fix $a_i\in\Rea$ and $b_i\ge\beta$, and set $\bar\alpha:= \max\lvert a_i\rvert$.
		Moreover, let $c_{1},\dots,c_{k}$ be independent random variables, uniformly distributed on
		the interval $[0,\ell]$. Let us consider the random variable
        $\hat\varphi:=\varphi+R$, where $R(x):=\sum\limits_{i=1}^{k}a_i\psi(b_i(x-c_i))$ for any $x\in \Rea$.
        If $\frac{\sigma}{\beta}< \frac{\ell}{8(k-1)}$, then
		$$E\left(\left\|
		F^{\frac{\sigma}{\beta}}\circ F_{2\frac{\sigma}{\beta}}(\hat\varphi)-\varphi
		\right\|_\infty\right)\le 3 L\frac{\sigma}{\beta}+k\bar\alpha\left(1-\left(1-8\frac{(k-1)}{\ell}\frac{\sigma}{\beta}\right)^{k}\right).$$
	\end{theorem}

	\begin{proof}
	    By setting $\delta=\frac{\sigma}{\beta}$ and $\varepsilon=2\frac{\sigma}{\beta}$ in statement $\mathrm{iii)}$ of Lemma~\ref{lemma1}, we have that $$\left\|
		F^{\frac{\sigma}{\beta}}\circ F_{2\frac{\sigma}{\beta}}(\hat\varphi)-\varphi
		\right\|_\infty\leq3L\frac{\sigma}{\beta}+\left\|
		F^{\frac{\sigma}{\beta}}\circ F_{2\frac{\sigma}{\beta}}(R)
		\right\|_\infty.$$ Since the operator $F^{\frac{\sigma}{\beta}}\circ F_{2\frac{\sigma}{\beta}}$ is non-expansive
        and $F^{\frac{\sigma}{\beta}}\circ F_{2\frac{\sigma}{\beta}}(\mathbf{0})=\mathbf{0}$,  it follows that
        $$\left\|
		F^{\frac{\sigma}{\beta}}\circ F_{2\frac{\sigma}{\beta}}(R)
		\right\|_\infty\leq\left\|R
		\right\|_\infty\leq k\bar\alpha\left\|\psi
		\right\|_\infty\leq k\bar\alpha.$$
		Therefore, $\left\|
		F^{\frac{\sigma}{\beta}}\circ F_{2\frac{\sigma}{\beta}}(\hat\varphi)-\varphi
		\right\|_\infty\leq3L\frac{\sigma}{\beta}+k\bar\alpha$.
If we apply Lemma~\ref{lemma3} with $\eta=8\frac{\sigma}{\beta}$, we obtain that ${R}\in\mathcal{R}_{8\frac{\sigma}{\beta},\beta}$ with probability $p:=\left(1-8\frac{(k-1)}{\ell}\frac{\sigma}{\beta}\right)^{k}$.
If ${R}\in\mathcal{R}_{8\frac{\sigma}{\beta},\beta}$, we can apply Theorem~\ref{Upper Bound} by setting $\delta=\frac{\sigma}{\beta}$, $\varepsilon=2\frac{\sigma}{\beta}$ and $\theta=8\frac{\sigma}{\beta}$, and hence we obtain that
        $\left\|
		F^{\frac{\sigma}{\beta}}\circ F_{2\frac{\sigma}{\beta}}(\hat\varphi)-\varphi
		\right\|_\infty\le 3 L\frac{\sigma}{\beta}$ with probability at least $p$.
Since $\left\|
		F^{\frac{\sigma}{\beta}}\circ F_{2\frac{\sigma}{\beta}}(\hat\varphi)-\varphi
		\right\|_\infty\le 3 L\frac{\sigma}{\beta}+k\bar\alpha$ in any case,
it follows that
\begin{align*}
            E\left(\left\|F^{\frac{\sigma}{\beta}}\circ F_{2\frac{\sigma}{\beta}}(\hat\varphi)-\varphi
		\right\|_\infty\right)&\le 3 L\frac{\sigma}{\beta}p+\left(3 L\frac{\sigma}{\beta}+k\bar\alpha\right)\left(1-p\right)\\
            &= 3 L\frac{\sigma}{\beta}+k\bar\alpha\left(1-p\right).
\end{align*}
\end{proof}
	
	
%
%
	
Theorem~\ref{main_result} and the well known stability of persistence diagrams with respect to the max-norm \cite{CSEdHa07} immediately imply the following result, which is of interest in TDA (the symbol $\dmatch$ denotes the usual bottleneck distance between persistence diagrams).
	
	\begin{corollary}
		\label{main_corollary}
Let us make the same assumptions of Theorem~\ref{main_result}.
Let $D$ and $D'$ be the persistence diagrams in degree $0$ of the filtering functions $\varphi$ and $F^{\frac{\sigma}{\beta}}\circ F_{2\frac{\sigma}{\beta}}(\hat\varphi)$,
respectively.  If $\frac{\sigma}{\beta}< \frac{\ell}{8(k-1)}$, then
		$$E\left(\dmatch(D,D')\right)\le 3 L\frac{\sigma}{\beta}+k\bar\alpha\left(1-\left(1-8\frac{(k-1)}{\ell}\frac{\sigma}{\beta}\right)^{k}\right).$$
	\end{corollary}
	
	This result shows that the use of suitable GENEOs can make TDA (relatively) stable also in the presence of impulsive noise, under the assumptions we have considered in this paper.

	\section{Examples and experiments}
	
	We will now validate our approach based on GENEOs by giving two examples and illustrating some experimental results.
	
	\subsection{Examples}
	In order to verify how our approach works, we will set $\tau_n:=\left(1-\frac{1}{n}\right)2\frac{\sigma}{\beta}+\frac{1}{n}\left(\frac{\theta}{2}-2\frac{\sigma}{\beta}\right)$ and consider the upper bound
	$\left\|F^{\frac{\tau_n}{2}} \circ F_{\tau_n} \left(\hat\varphi\right)-\varphi\right\|_\infty\leq \frac{3}{2}L\tau_n$, obtained by applying Corollary~\ref{coroll2}. We observe that $\tau_n\ge 2\frac{\sigma}{\beta}$ for every index $n$, and $\lim\limits_{n\to+\infty}\tau_n=2\frac{\sigma}{\beta}$.
	We will examine two examples that use the GENEOs $F^{\frac{\tau_n}{2}} \circ F_{\tau_n}$ and show how our method based on such operators and the method based on convolutions differ, as for their capability in removing additive impulsive noise. Moreover, we will compare the actual error $\left\|F^{\frac{\tau_n}{2}} \circ F_{\tau_n} \left(\hat\varphi\right)-\varphi\right\|_\infty$ to its upper bound $ \frac{3}{2}L\tau_n$, by running several simulations.
	The convolutions that will be applied in our examples use the functions $T_h:\Rea\to\Rea$ defined by setting $$T_h(x):=
	\begin{cases}
		\frac{h}{2} &\mbox{if } -\frac{1}{h}\le x\le\frac{1}{h}\\
		0 &\mbox{otherwise}
	\end{cases}$$ for $h>0$.
	We will see that, although the convolution with such functions is also a GENEO, it will not be able to efficiently remove the noise.
	
	Our noise function $R$ will be constructed starting from the mother function $\psi$ defined by setting $\psi(x):=e^{1-\frac{1}{1-x^2}}$ for $x\in\left]-1,1\right[$ and $\psi(x):=0$ for  $x\notin\left]-1,1\right[$. Using the notation introduced in the previous sections, we will set $\sigma=1.1$, thus satisfying the condition $\mathrm{supp}(\psi)\subseteq\left]-\sigma,\sigma\right[$.
	The impulsive noise will be added in an interval $\left[-{{\ell}},{{\ell}}\ \right]$.
	The following parameters are considered, with these respective uniform distributions:
	\begin{itemize}
		\item $k \sim \mathrm{Unif}_{\{1,\ldots,10\}}$
		\item $a_i \sim \mathrm{Unif}_{\left]-100,100\right[}$ for $i=1,\ldots,k$
		\item $b_i \sim \mathrm{Unif}_{\left]0,100\right[}$ for $i=1,\ldots,k$
		\item $c_i \sim \mathrm{Unif}_{\left]-{{\ell}}+\frac{\sigma}{\beta},{{\ell}}-\frac{\sigma}{\beta}\right[}$ for $i=1,\ldots,k$.
	\end{itemize}
	We set $\beta:=\min\limits_{i=1,\ldots,N} b_i$, $\bar\alpha:=\max\limits_{i=1,\ldots,N} \lvert a_i\rvert$, and $\eta:=\min\limits_{i\not=j} \lvert c_i-c_j\rvert$.
	After producing random values for the parameters $N,a_i,b_i,c_i$, our algorithm checks whether $\eta>8\frac{\sigma}{\beta}$, otherwise it generates another set of parameters.
	
	\subsubsection{First example}
	Let us consider the function
	$$\varphi(x):=
	\begin{cases}
		\sin x &\mbox{if } -4\pi\le x\le4\pi \\
		0 &\mbox{otherwise}
	\end{cases}$$ for $x\in\Rea$. We observe that $\varphi\in\LipL$, for $L=1$. We will add noise in the interval $[-4\pi,4\pi]$ (meaning ${{\ell}}=4\pi$) and visualize the results in such an interval.
	Figure~\ref{fig_ex_1} illustrates how the function $\hat\varphi$ looks like compared to $\varphi$.
	\begin{figure}[t]
		\centering
		\includegraphics[width=4.8cm]{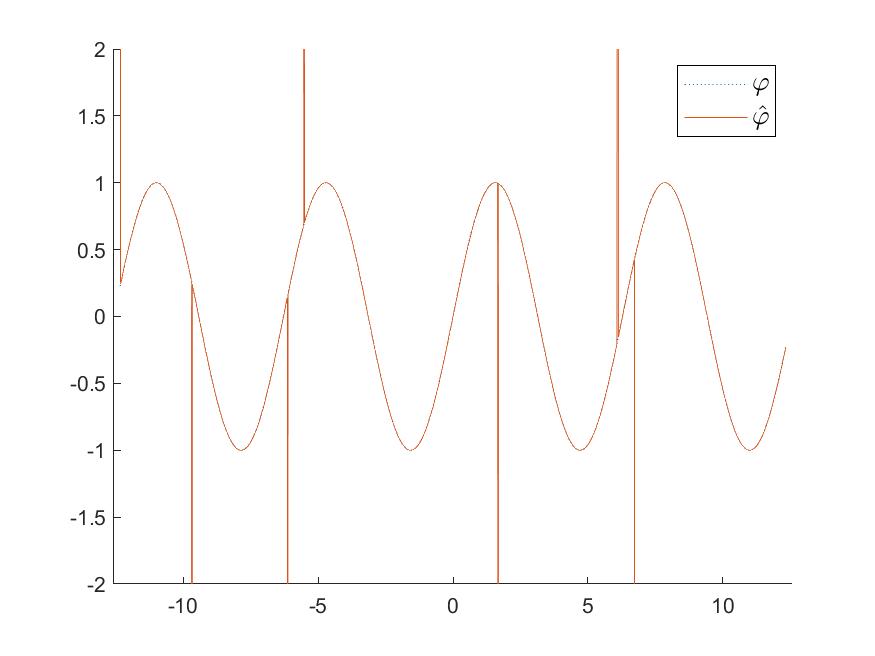}
		\caption{Example 1: Comparison between $\varphi$ and $\hat\varphi$.}
		\label{fig_ex_1}
	\end{figure}
	
	We will start by considering how well the convolution $\hat\varphi*T_n$ can approximate the original function $\varphi$, when $n$ goes from $3$ to $100$.
	From Figure~\ref{table_ex_1_A}, it is immediately apparent that the max-norm distance between $\hat\varphi*T_n$ and $\varphi$ remains quite large.
	\begin{figure}[t]
		\begin{center}
			\begin{tabular}{r c}	\includegraphics[width=4.8cm]{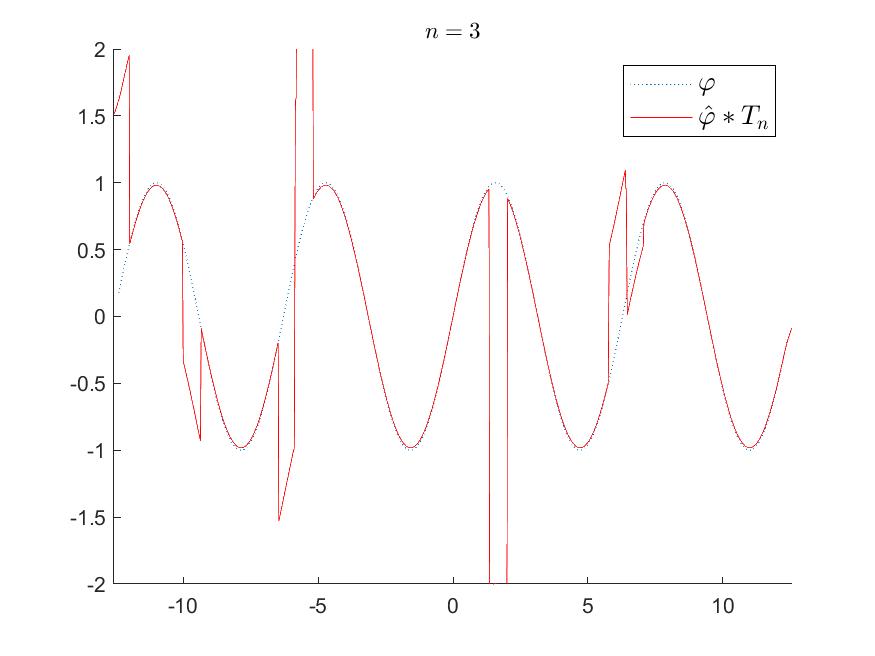} & \includegraphics[width=4.8cm]{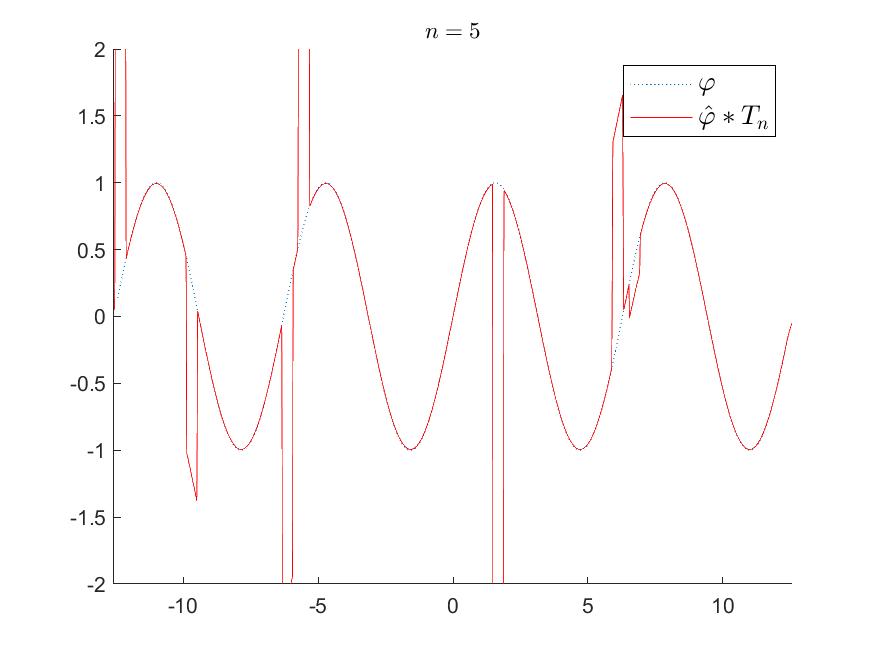} \\
			\includegraphics[width=4.8cm]{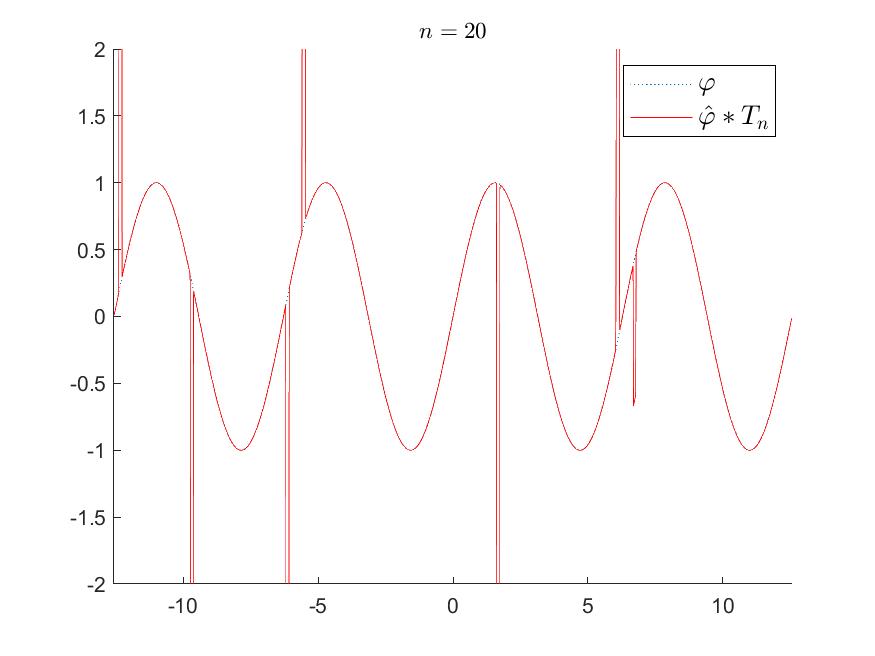} & \includegraphics[width=4.8cm]{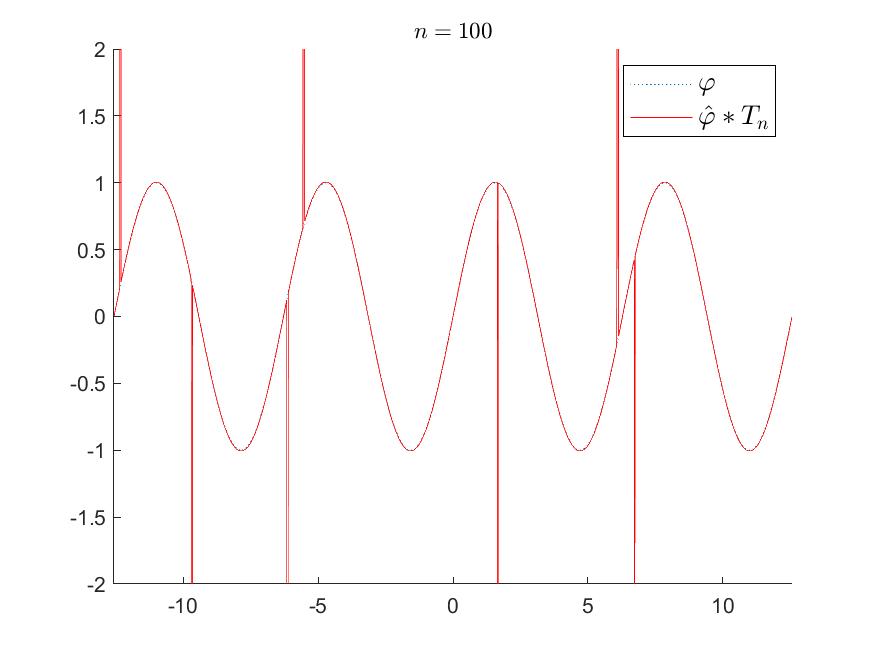}  \\
			\end{tabular}
			\caption{Example 1: Denoising via convolution with $T_h$ ($h=n$).}
			\label{table_ex_1_A}
		\end{center}
	\end{figure}
	If we apply a convolution with $T_\frac{1}{n}$, for $3\le n\le 100$, we get the results displayed in Figure~\ref{table_ex_1_B}, showing that all information represented by the function $\varphi$ is progressively destroyed.
	\begin{figure}[t]
		\begin{center}
			\begin{tabular}{r c}
			\includegraphics[width=4.8cm]{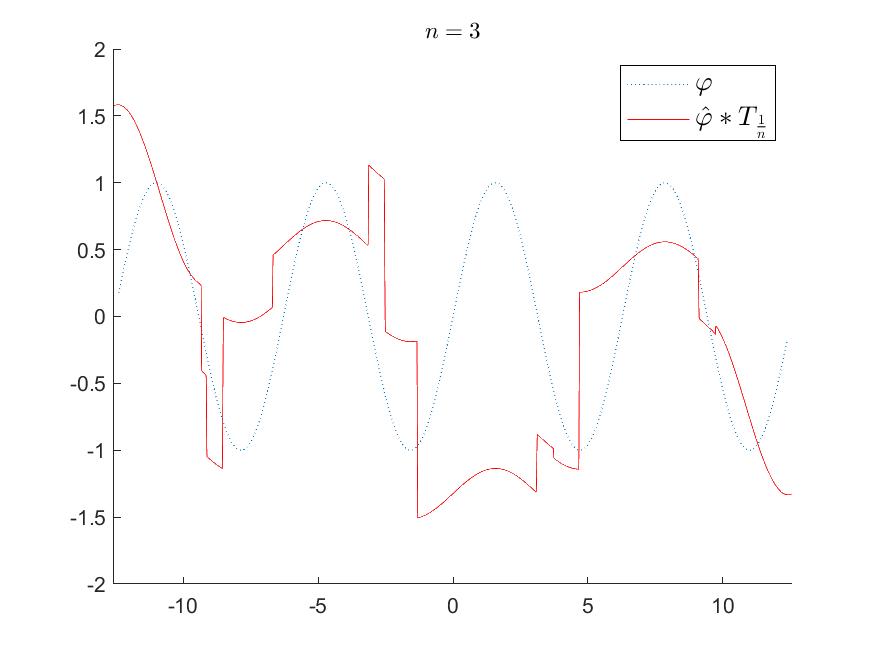} & \includegraphics[width=4.8cm]{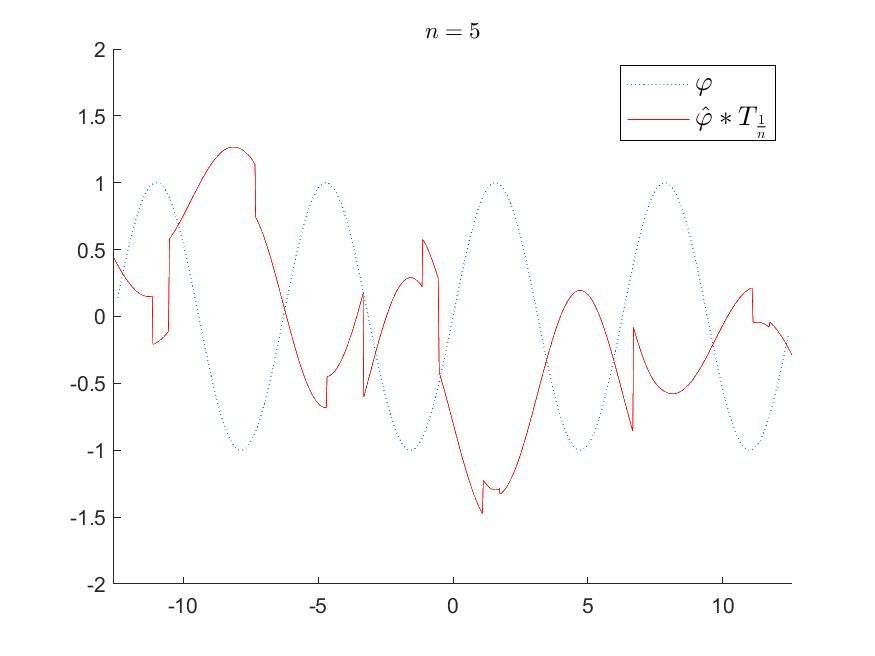} \\
			\includegraphics[width=4.8cm]{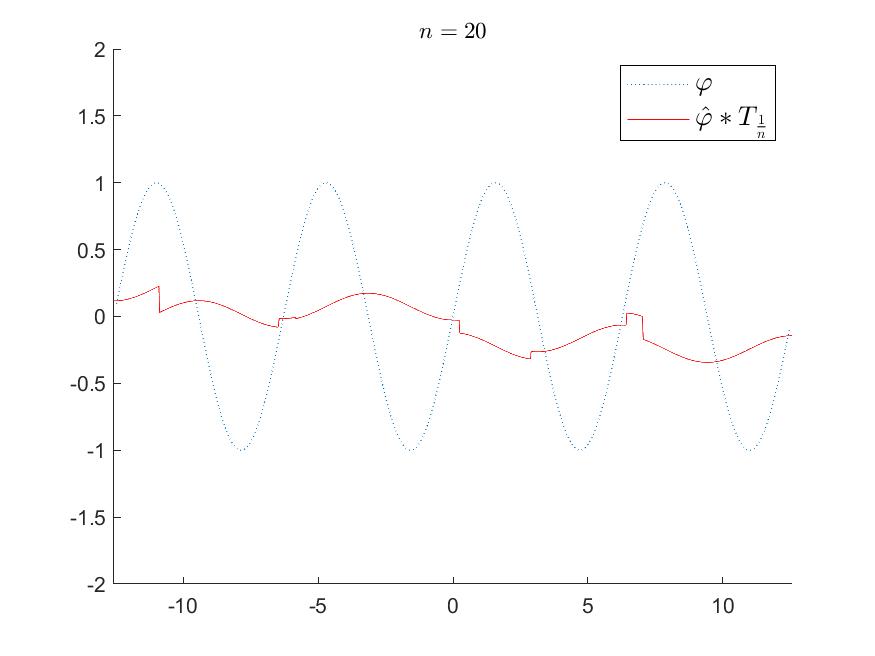} & \includegraphics[width=4.8cm]{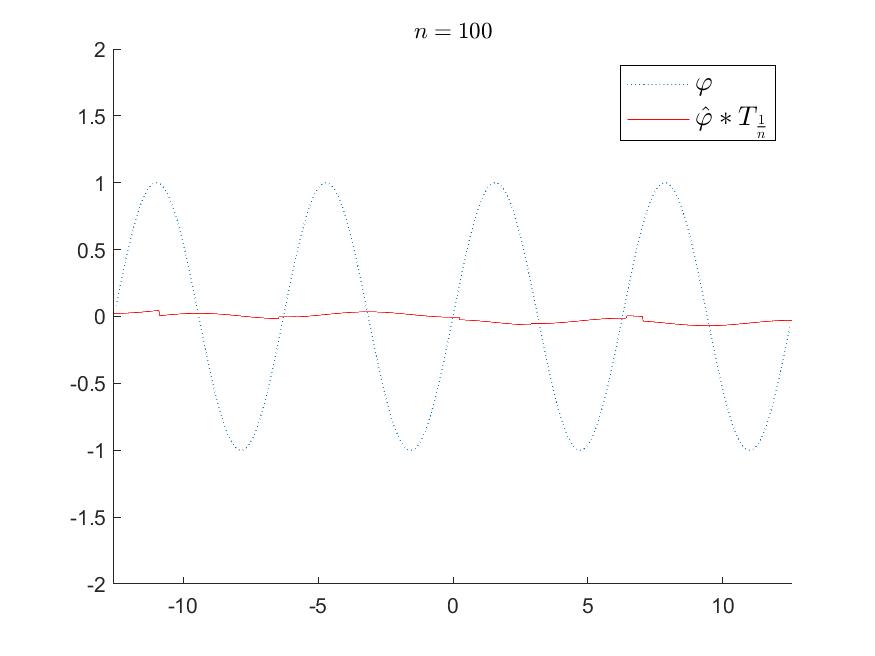}  \\
			\end{tabular}
			\caption{Example 1: Denoising via convolution with $T_h$ ($h=\frac{1}{n}$).}
			\label{table_ex_1_B}
		\end{center}
	\end{figure}
	In contrast, if we apply the operator $F^{\frac{\tau_n}{2}} \circ F_{\tau_n} \left(\hat\varphi\right)$, for $n=3,5,20,100$, we get the results displayed in Figure~\ref{table_ex_1_C}. As we can see, this operator is much more efficient in removing the bumps and restoring the function $\varphi$.
	
	\begin{figure}[t]
		\begin{center}
			\begin{tabular}{r c}
			\includegraphics[width=4.8cm]{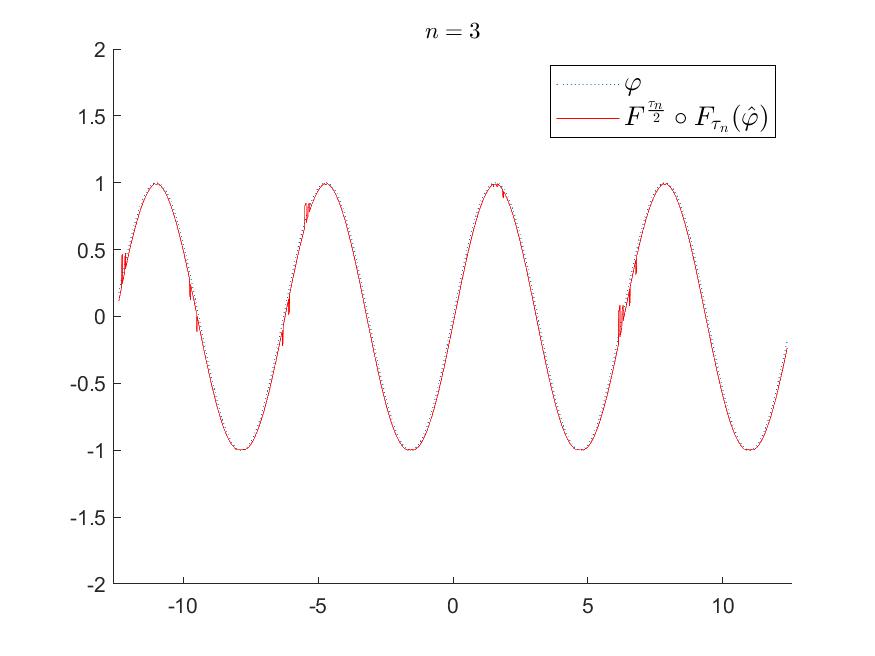} & \includegraphics[width=4.8cm]{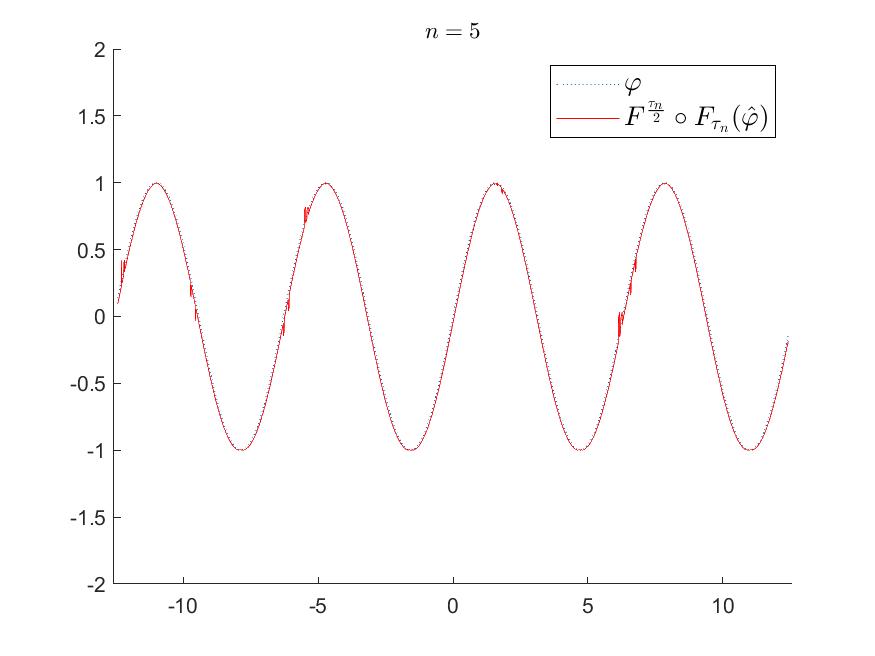} \\
			\includegraphics[width=4.8cm]{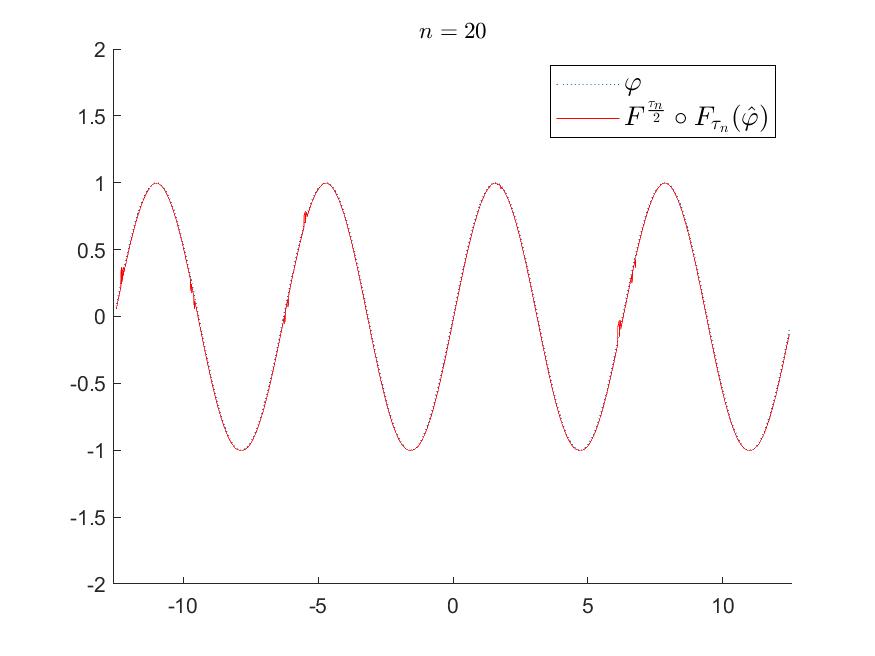} & \includegraphics[width=4.8cm]{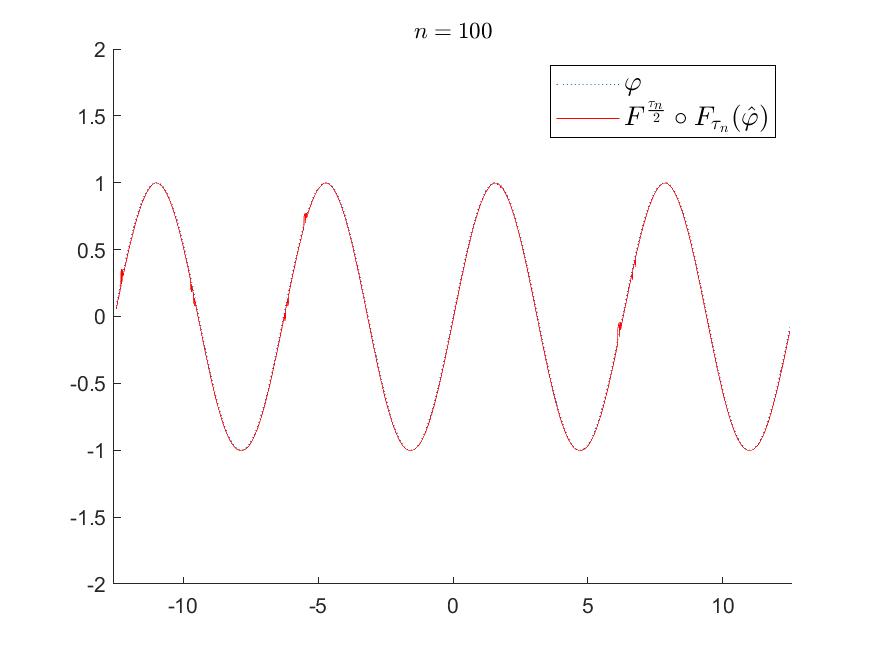}  \\
			\end{tabular}
			\caption{Example 1: Denoising via the proposed GENEO $F^{\frac{\tau_n}{2}} \circ F_{\tau_n}$.}
			\label{table_ex_1_C}
		\end{center}
	\end{figure}
	
	As a matter of fact, when we apply a convolution with the function $T_h$ and check the corresponding errors via the sup-norm, we get the results displayed in Figure~\ref{table_ex_1_D}.
	Since $\lim\limits_{h\to+\infty}\hat\varphi*T_{h}=\hat\varphi$ and $\lim\limits_{h\to+\infty}\hat\varphi*T_{\frac{1}{h}}=0$, we get that the errors tend to  $\|\hat\varphi-\varphi\|_\infty=\max\limits_{i=1,\ldots,N} \lvert a_i\rvert=\bar\alpha$ and $\|\varphi\|_\infty$, respectively.
	
	\begin{figure}[t]
		\begin{center}
			\begin{tabular}{c c}             \includegraphics[width=4.8cm]{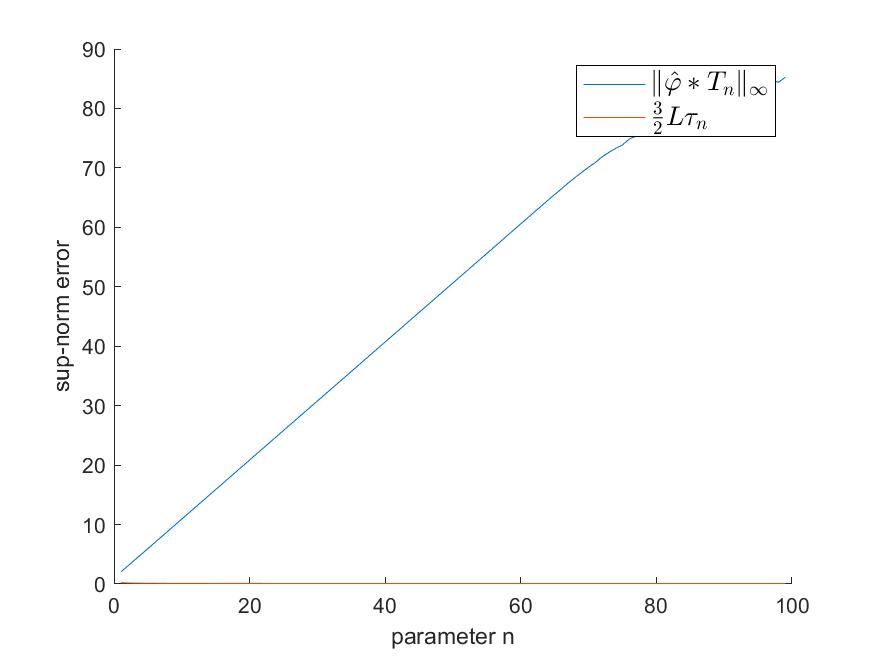} & \includegraphics[width=4.8cm]{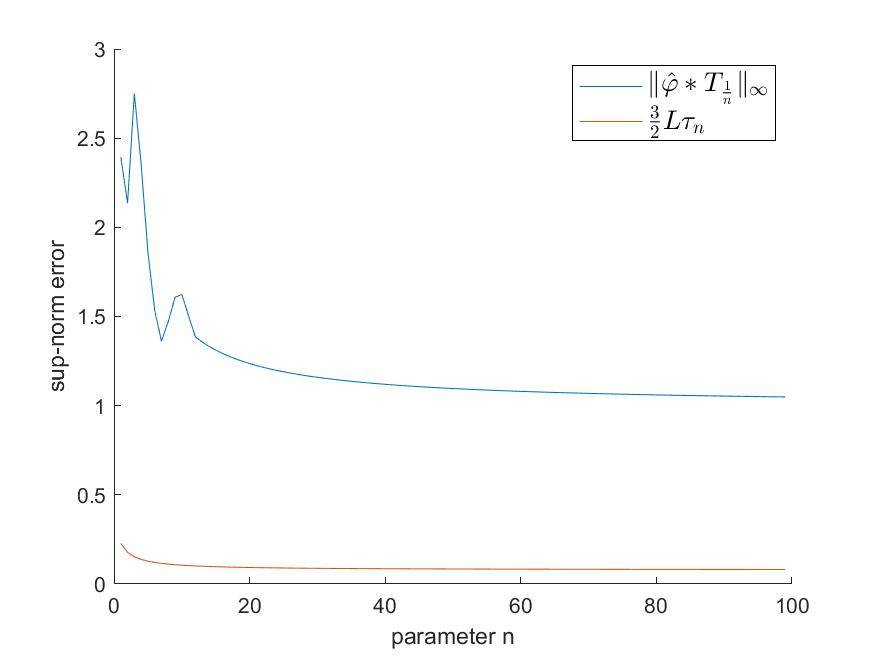}
			\end{tabular}
		\end{center}
		\caption{Example 1: Error made by applying a convolution with $T_n$ (left) or $T_\frac{1}{n}$ (right).}
		\label{table_ex_1_D}
	\end{figure}
	
	In contrast, if we apply the operator $F^{\frac{\tau_n}{2}} \circ F_{\tau_n}$, we get the results displayed in Figure~\ref{fig_ex_1_E}, showing that the upper bound for the error stated in Corollary~\ref{coroll2} is quite tight. As we can expect, we get the best denoising by replacing $\tau_n$ with $2\frac{\sigma}{\beta}$ (see Figure~\ref{fig_ex_1_F}).
	\begin{figure}[t]
		\centering
		\includegraphics[width=4.8cm]{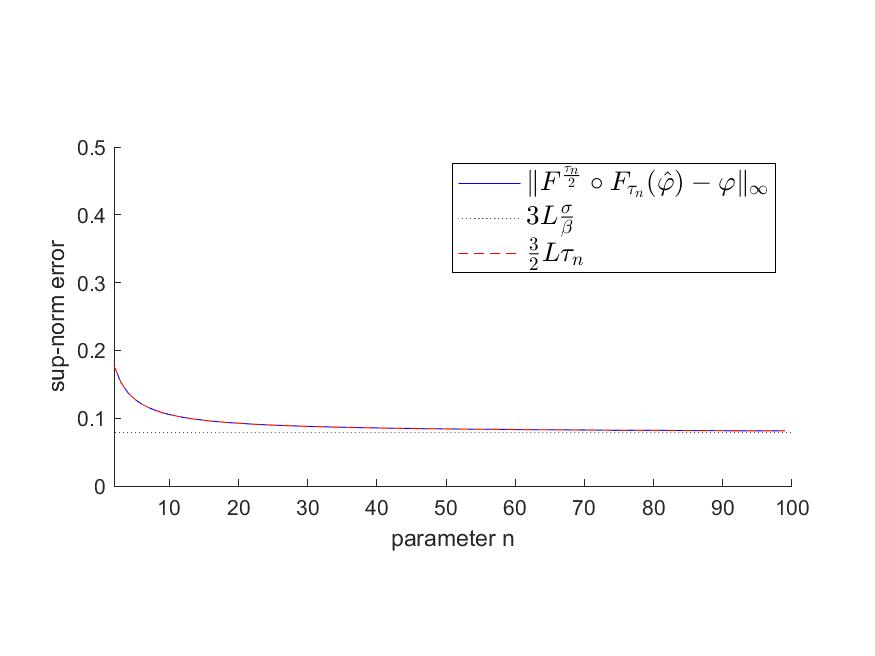}
		\caption{Example 1: Error made by using the GENEO $F^{\frac{\tau_n}{2}} \circ F_{\tau_n}$.}
		\label{fig_ex_1_E}
	\end{figure}
	
	\begin{figure}[t]
		\centering
		\includegraphics[width=4.8cm]{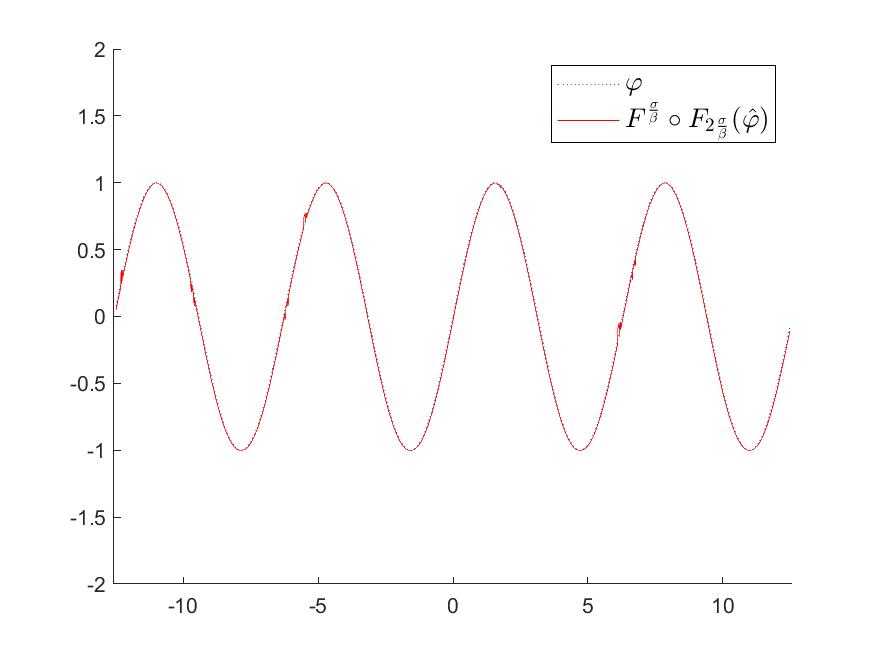}
		\caption{Example 1: Denoising via $F^\frac{\sigma}{\beta} \circ F_{2\frac{\sigma}{\beta}}$.}
		\label{fig_ex_1_F}
	\end{figure}
	
	Finally, we executed 1000 simulations. In each of them we have produced a function $\hat\varphi$ by adding  random impulsive noise to the function $\varphi$, then we have applied  $F^{\frac{\sigma}{\beta}} \circ F_{2\frac{\sigma}{\beta}}$ to $\hat\varphi$, in order to see how tight the upper bound in Corollary~\ref{coroll2} is to the actual error $\|F^{\frac{\sigma}{\beta}} \circ F_{2\frac{\sigma}{\beta}}(\hat\varphi)-\varphi\|_\infty$. The same parameters as in the beginning of this example have been used. As we can see in Figure~\ref{histo1} the overestimation committed by our upper bound is often quite close to zero, relatively to the Lipschitz constant $L$ of the function $\varphi$.
	\begin{figure}[t]
		\centering
		\includegraphics[width=4.8cm]{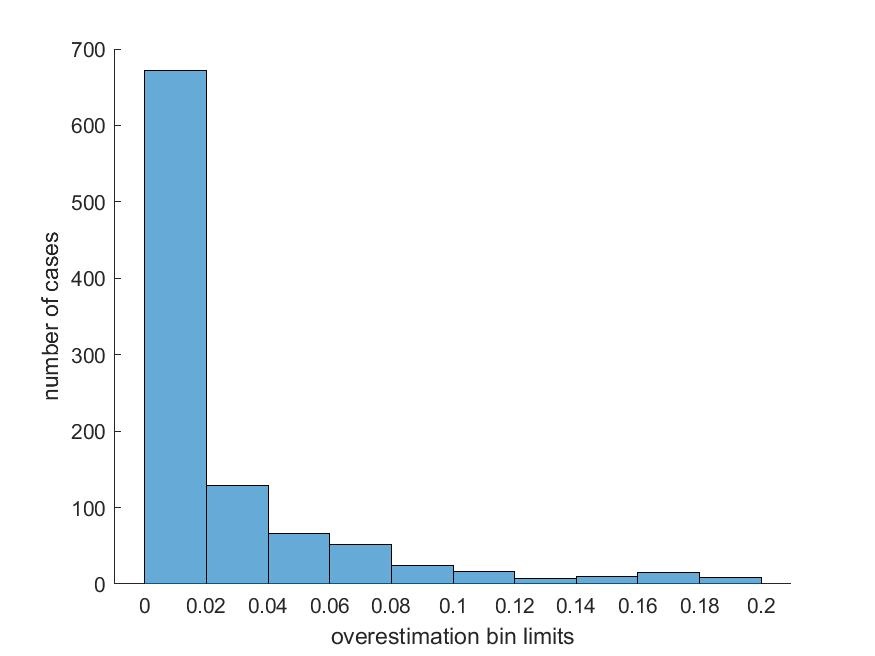}
		\caption{Example 1: Histogram counting the number of cases in each of ten bins, concerning the overestimation value $3L\frac{\sigma}{\beta}-\|F^{\frac{\sigma}{\beta}} \circ F_{2\frac{\sigma}{\beta}}(\hat\varphi)-\varphi\|_\infty$ obtained in our simulations. Most of the cases belong to the first bin, containing the functions $\varphi$ for which
$0\le 3L\frac{\sigma}{\beta}-\|F^{\frac{\sigma}{\beta}} \circ F_{2\frac{\sigma}{\beta}}(\hat\varphi)-\varphi\|_\infty\le 0.02$.}
		\label{histo1}
	\end{figure}
	This suggests that such an upper bound is quite accurate.
	
	\subsubsection{Second example}
	Let us consider the function $$\varphi(x):=
	\begin{cases}
		\frac{1}{1000}(x-5)(x-3)(x+1)(x+4)(x+5) &\mbox{if } -5\le x\le5 \\
		0 &\mbox{otherwise}
	\end{cases}$$ for $x\in\Rea$. The coefficient $\frac{1}{1000}$ was chosen in order to get that the Lipschitz constant $L$ is comparable to the one of the previous example. In this example $L=\frac{27}{25}$.
	Figure~\ref{fig_ex_2} illustrates how the function $\hat\varphi$ looks like compared to $\varphi$.
	\begin{figure}[t]
		\centering
		\includegraphics[width=4.8cm]{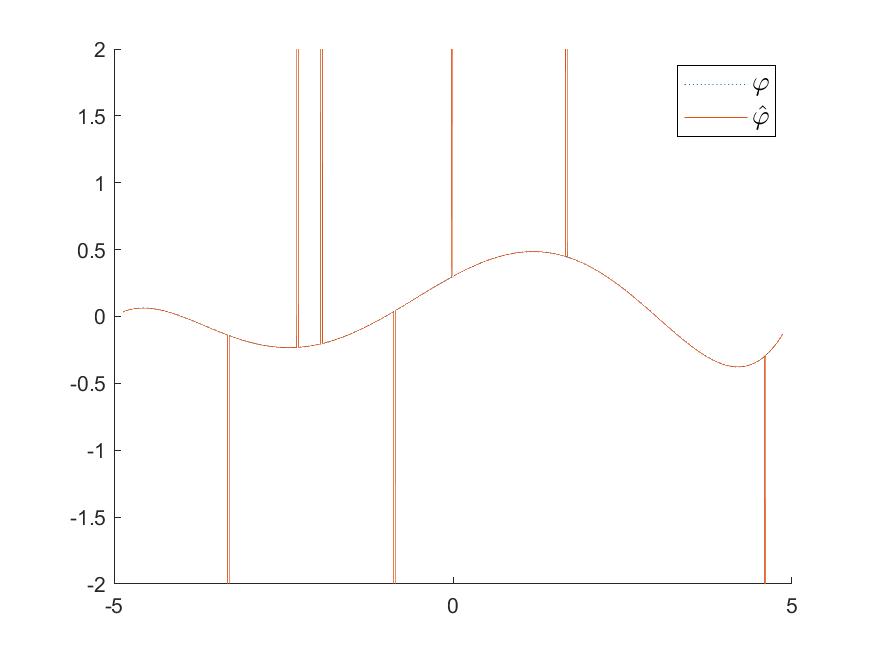}
		\caption{Example 2: Comparison between $\varphi$ and $\hat\varphi$.}
		\label{fig_ex_2}
	\end{figure}
	
	We will start by considering how well the convolution $\hat\varphi*T_n$ can approximate the original function $\varphi$, when $n$ goes from $3$ to $100$.
	From Figure~\ref{table_ex_2_A}, it is immediately apparent that the max-norm distance between $\hat\varphi*T_n$ and $\varphi$ remains quite large.
	\vskip 5cm
	
	\begin{figure}[t]
		\begin{center}
			\begin{tabular}{r c}
			\includegraphics[width=4.8cm]{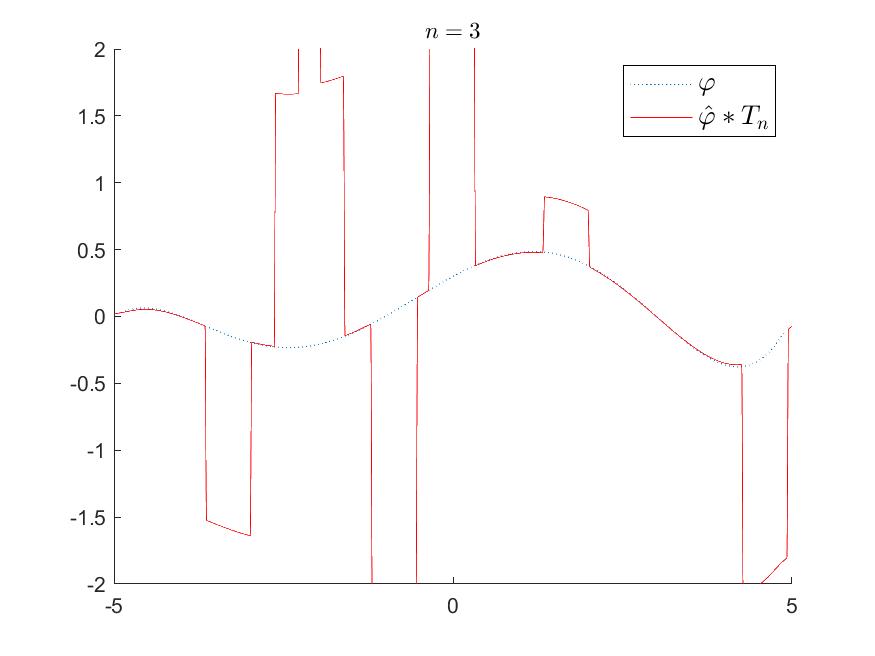} & \includegraphics[width=4.8cm]{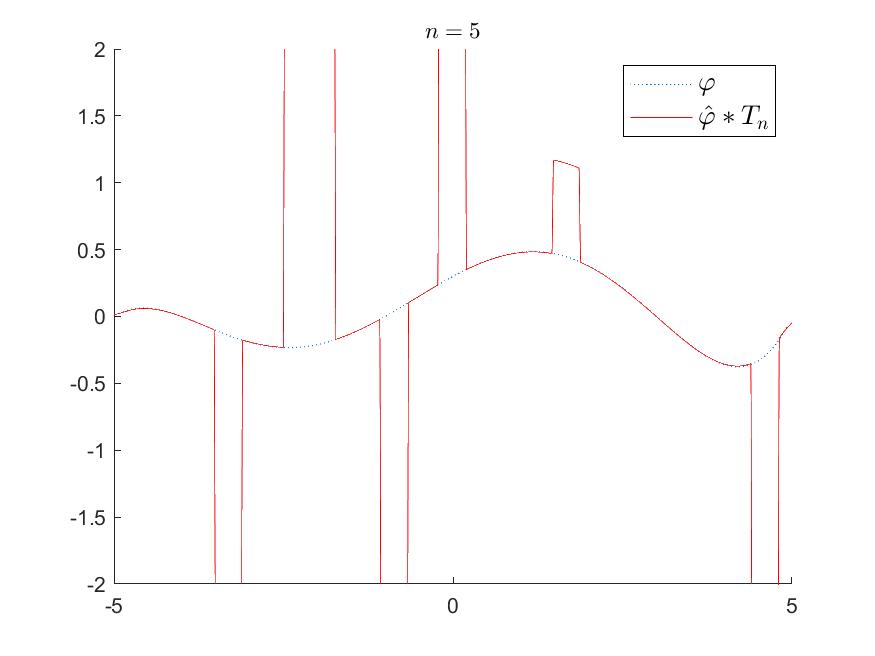} \\
			\includegraphics[width=4.8cm]{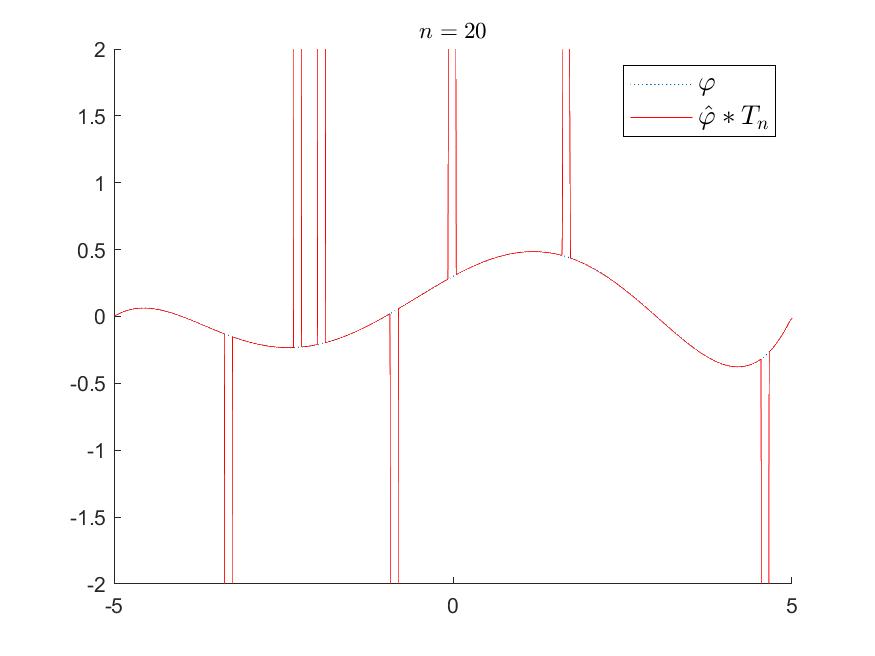} & \includegraphics[width=4.8cm]{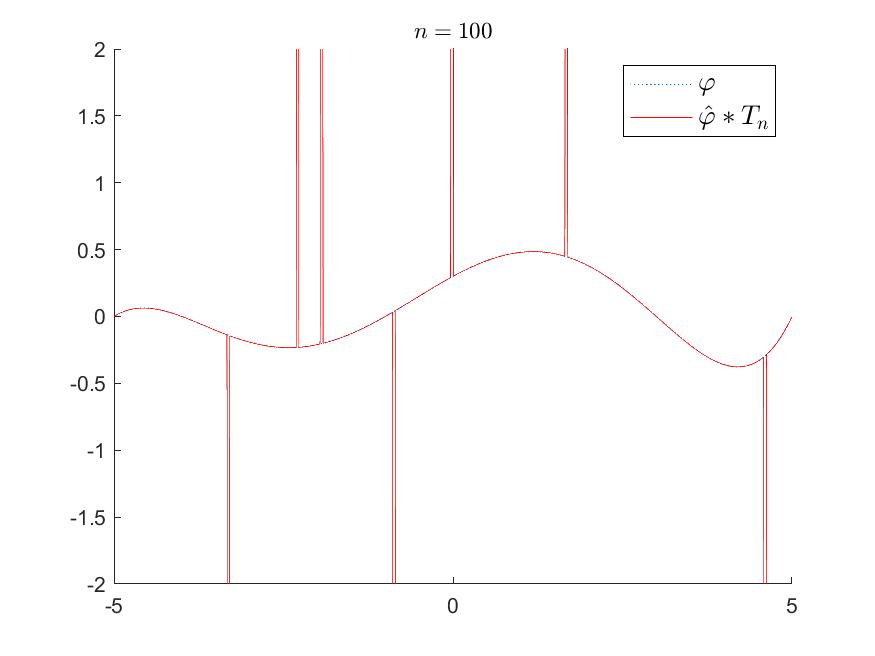}  \\
			\end{tabular}
			\caption{Example 2: Denoising via convolution with $T_h$ ($h=n$).}
			\label{table_ex_2_A}
		\end{center}
	\end{figure}
	
	If we apply a convolution with $T_\frac{1}{n}$, for $3\le n\le 100$, we get the results displayed in Figure~\ref{table_ex_2_B}, showing that all information represented by the function $\varphi$ is progressively destroyed.
	\begin{figure}[t]
		\begin{center}
			\begin{tabular}{r c}
			\includegraphics[width=4.8cm]{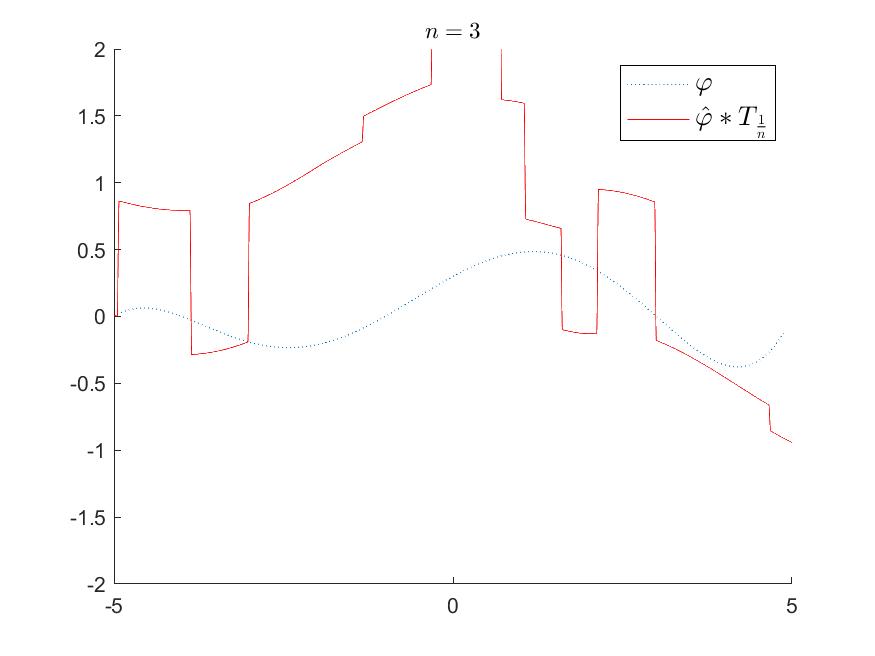} & \includegraphics[width=4.8cm]{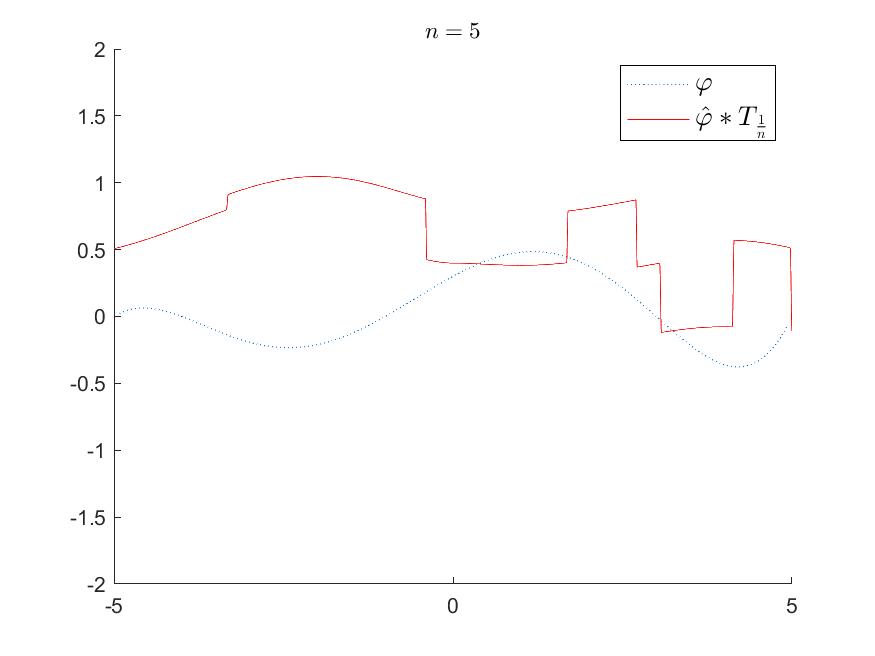} \\
			\includegraphics[width=4.8cm]{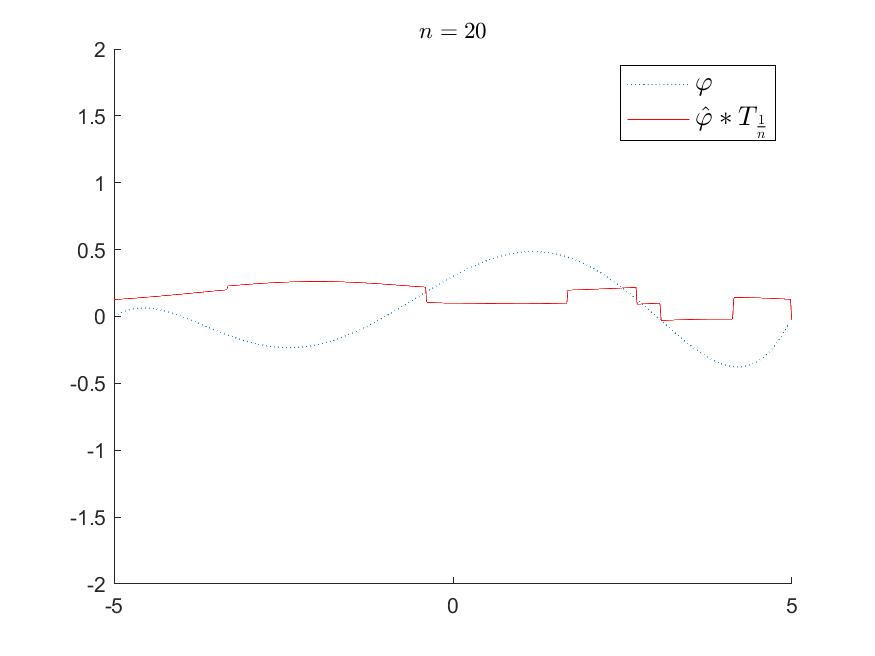} & \includegraphics[width=4.8cm]{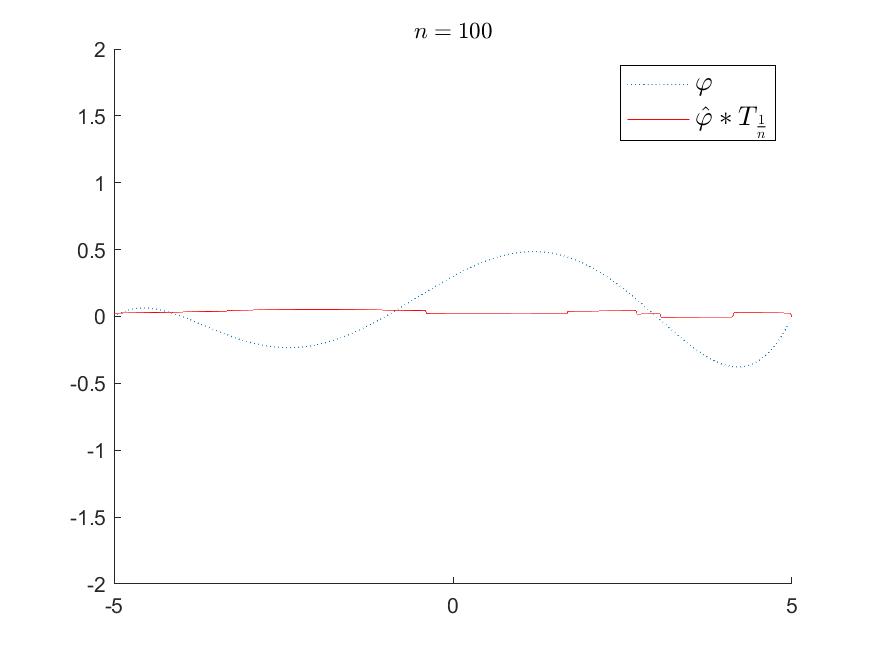}  \\
			\end{tabular}
			\caption{Example 2: Denoising via convolution with $T_h$ ($h=\frac{1}{n}$).}
			\label{table_ex_2_B}
		\end{center}
	\end{figure}
	In contrast, if we apply the operator $F^{\frac{\tau_n}{2}} \circ F_{\tau_n} \left(\hat\varphi\right)$, for $n=3,5,20,100$, we get the results displayed in Figure~\ref{table_ex_2_C}. As we can see, this operator is much more efficient in removing the bumps and restoring the function $\varphi$.
	
	When we apply a convolution with the function $T_h$ and check the corresponding errors via the sup-norm, we get the results displayed in Figure~\ref{table_ex_2_D}.
	As already seen, since $\lim\limits_{h\to+\infty}\hat\varphi*T_{h}=\hat\varphi$ and $\lim\limits_{h\to+\infty}\hat\varphi*T_{\frac{1}{h}}=0$ we get that the errors tend to  $\|\hat\varphi-\varphi\|_\infty=\max\limits_{i=1,\dots,k} \lvert a_i\rvert=\bar\alpha$ and $\|\varphi\|_\infty$, respectively.
	
	\begin{figure}[t]
		\begin{center}
			\begin{tabular}{r c}
			\includegraphics[width=4.8cm]{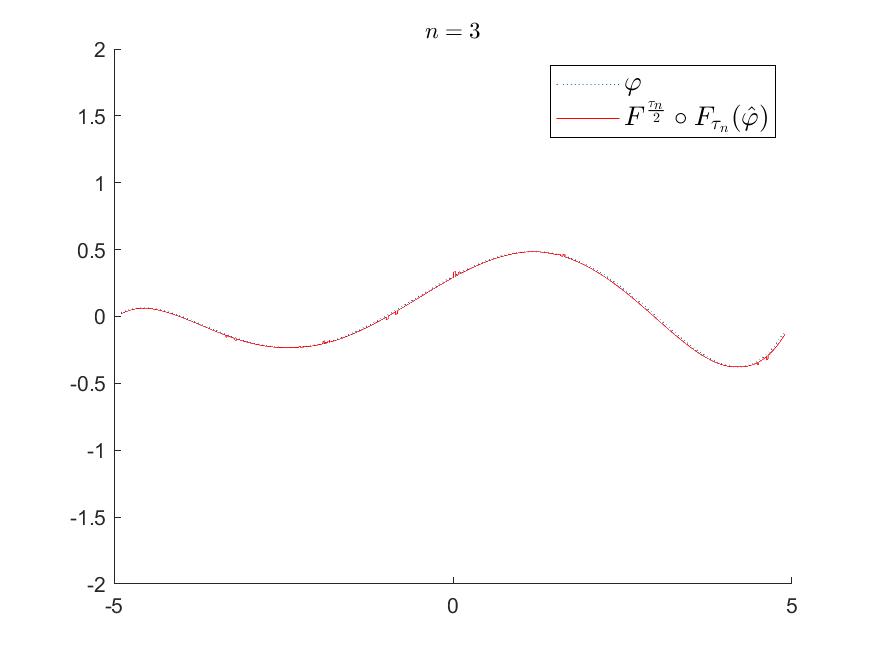} & \includegraphics[width=4.8cm]{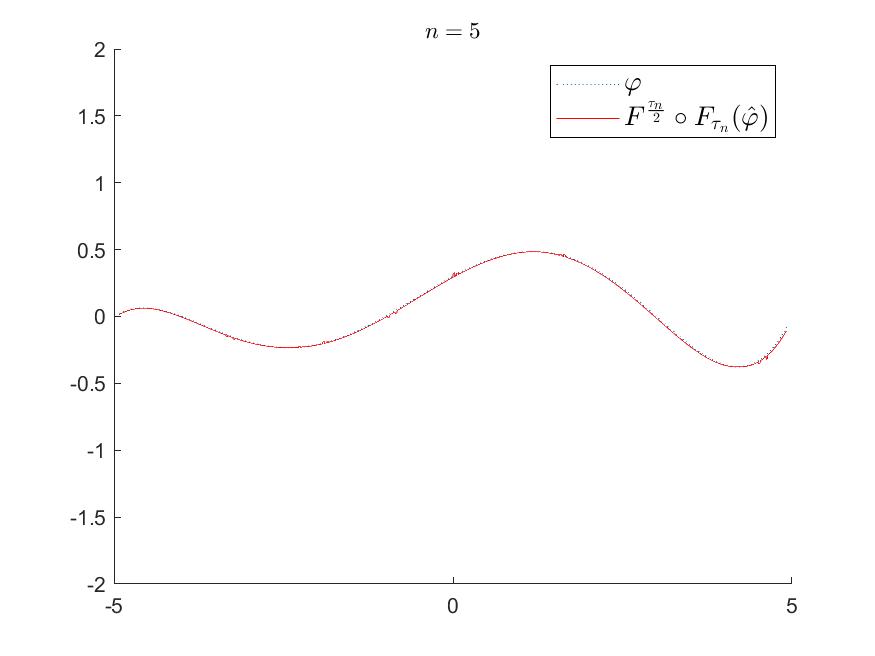} \\
			\includegraphics[width=4.8cm]{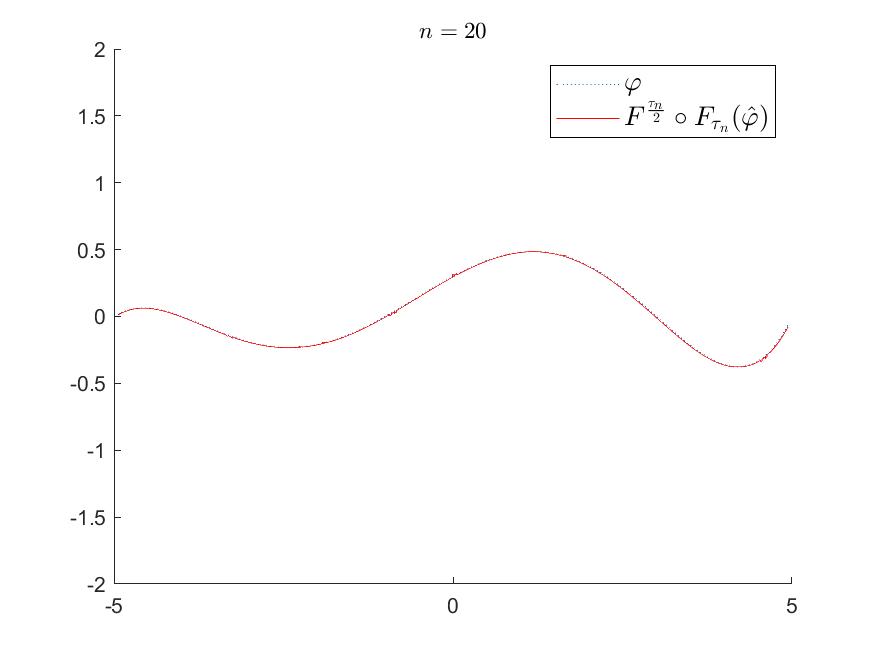} & \includegraphics[width=4.8cm]{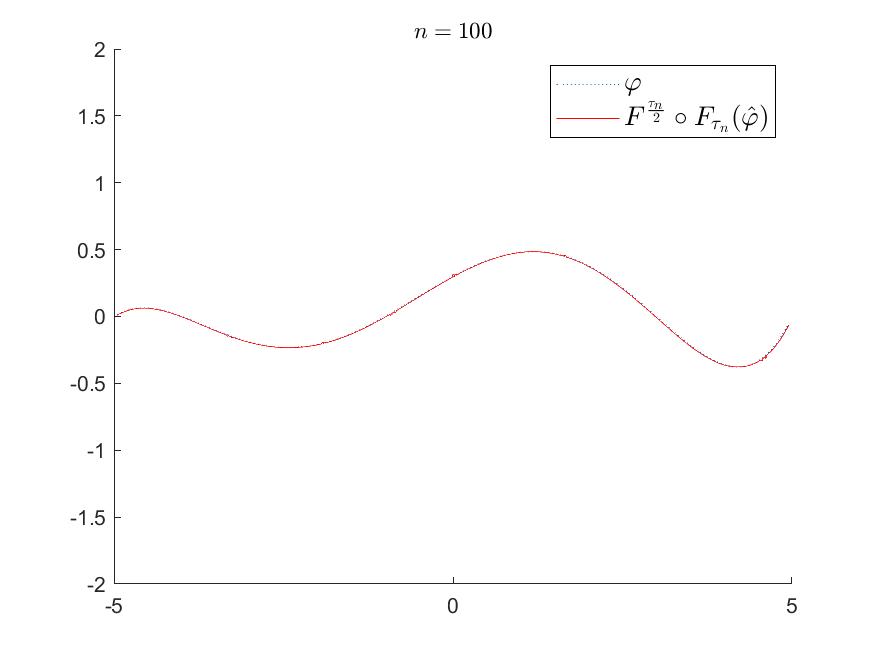}  \\
			\end{tabular}
			\caption{Example 2: Denoising via the proposed GENEO $F^{\frac{\tau_n}{2}} \circ F_{\tau_n}$.}
			\label{table_ex_2_C}
		\end{center}
	\end{figure}
	
	\begin{figure}[t]
		\begin{center}
			\begin{tabular}{c c}
	   	       \includegraphics[width=4.8cm]{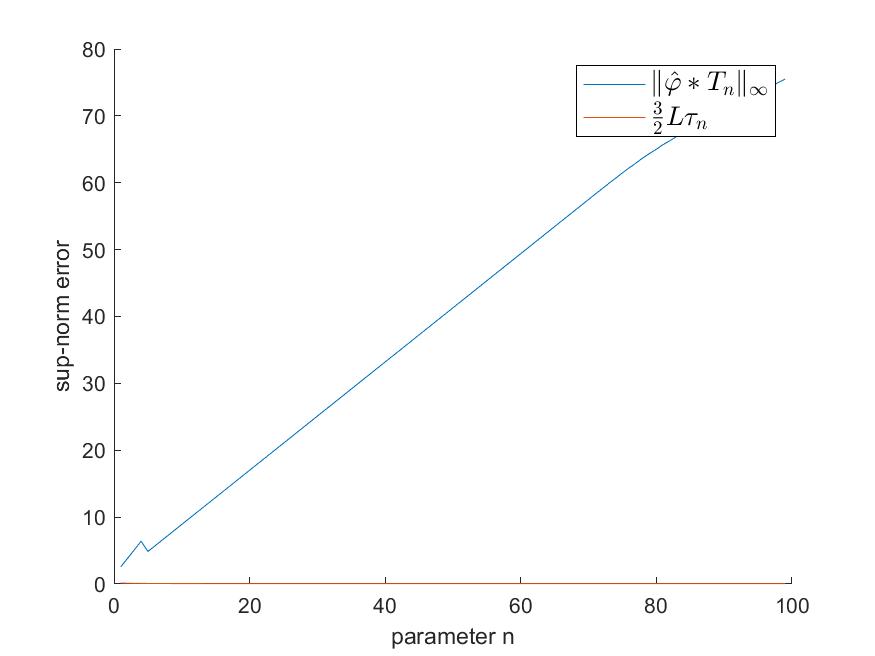} &
			   \includegraphics[width=4.8cm]{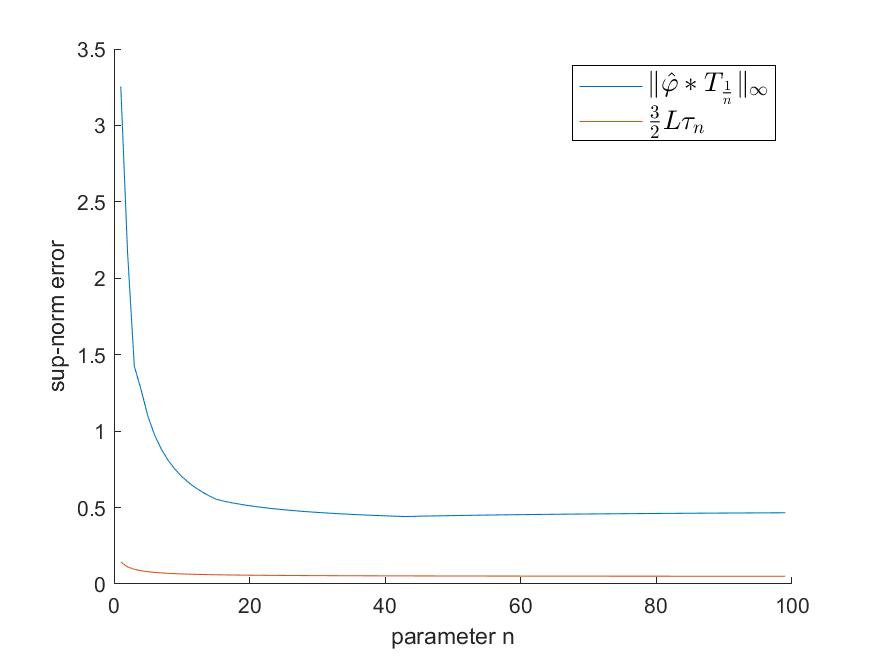}
			\end{tabular}
		\end{center}
		\caption{Example 2: Error made by applying a convolution with $T_n$ (left) or $T_\frac{1}{n}$ (right).}
		\label{table_ex_2_D}
	\end{figure}
	
	In contrast, if we apply the operator $F^{\frac{\tau_n}{2}} \circ F_{\tau_n}$, we get the results displayed in Figure~\ref{fig_ex_2_E}, showing that the upper bound for the error stated in Corollary~\ref{coroll2} is quite tight. As we can expect, we get the best denoising by replacing $\tau_n$ with $2\frac{\sigma}{\beta}$ (see Figure~\ref{fig_ex_2_F}).
	\begin{figure}[t]
		\centering
			\includegraphics[width=4.8cm]{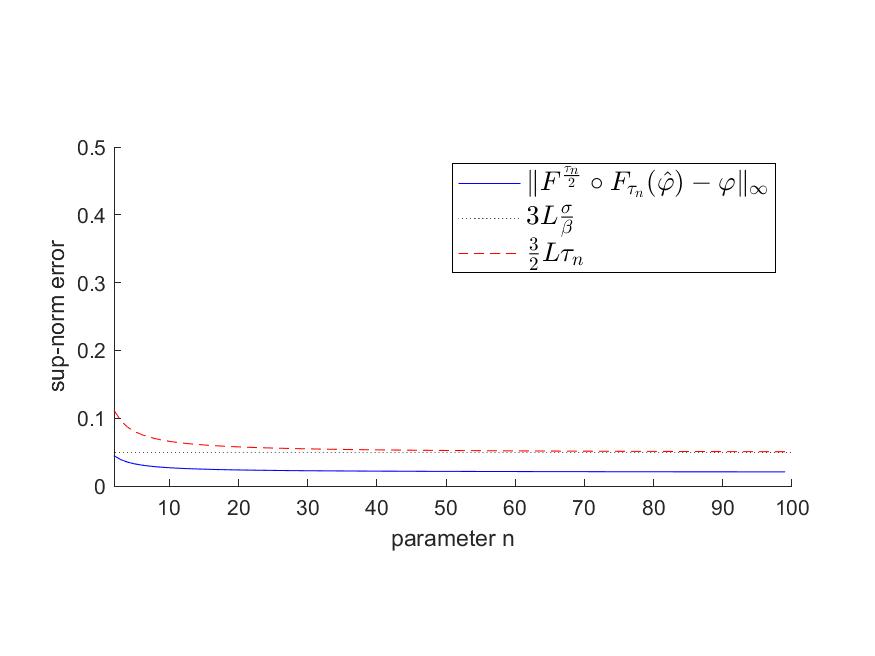}
		\caption{Example 2: Error made by using the GENEO $F^{\frac{\tau_n}{2}} \circ F_{\tau_n}$.}
		\label{fig_ex_2_E}
	\end{figure}
	
	\begin{figure}[t]
		\centering
		\includegraphics[width=4.8cm]{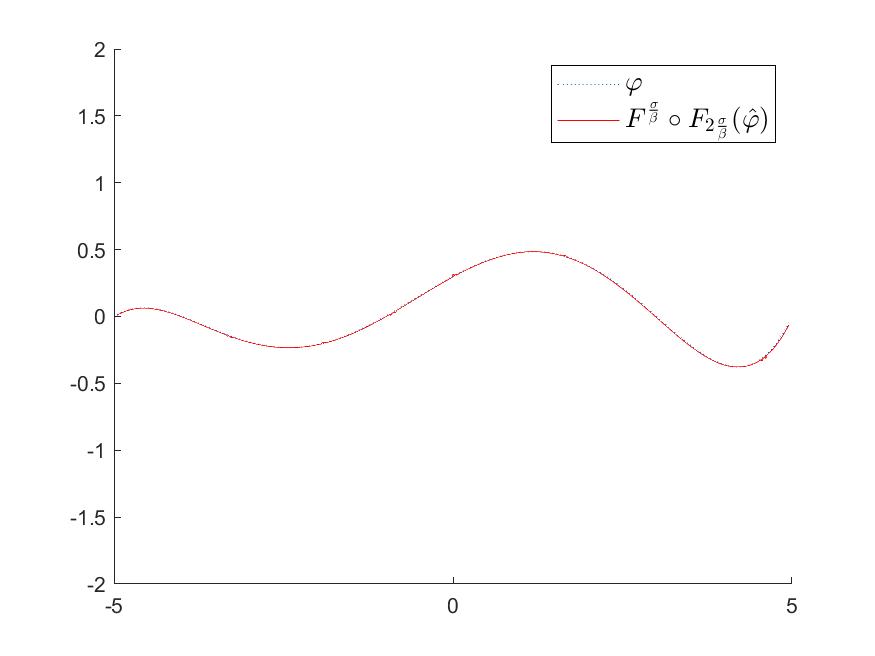}
		\caption{Example 2: Denoising via $F^\frac{\sigma}{\beta} \circ F_{2\frac{\sigma}{\beta}}$.}
		\label{fig_ex_2_F}
	\end{figure}
	
	Finally, we again executed 1000 simulations, using the same methodology as in the previous case, this time considering the polynomial presented at the beginning of this example. As we can see in Figure~\ref{histo2} the overestimation committed by our upper bound is often quite close to zero, relatively to the Lipschitz constant $L$ of the polynomial $\varphi$.
	\begin{figure}[t]
		\centering
		\includegraphics[width=4.8cm]{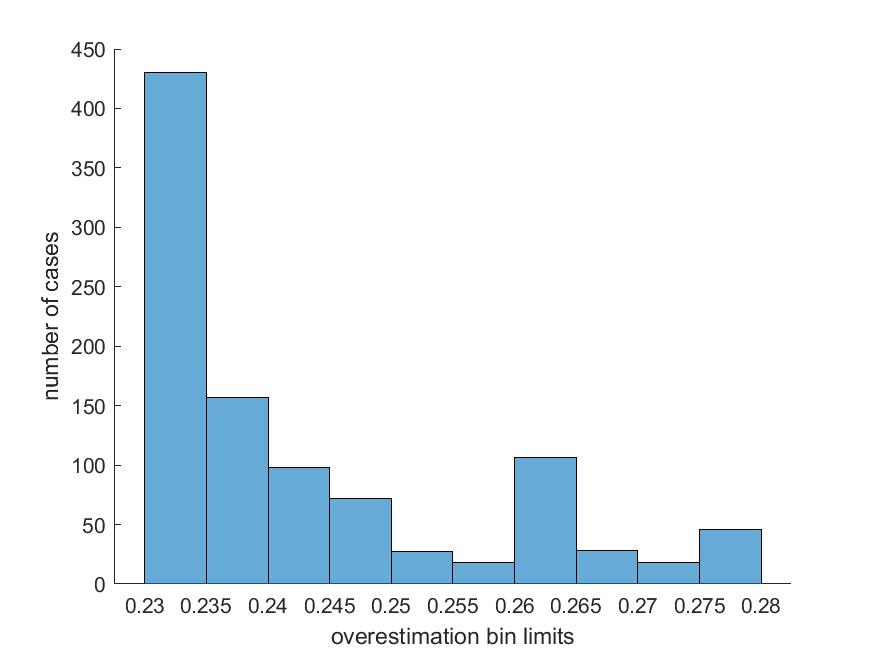}
		\caption{Example 2: Histogram counting the number of cases in each of ten bins, concerning the overestimation value $3L\frac{\sigma}{\beta}-\|F^{\frac{\sigma}{\beta}} \circ F_{2\frac{\sigma}{\beta}}(\hat\varphi)-\varphi\|_\infty$ obtained in our simulations. Most of the cases belong to the first two bins, containing the functions $\varphi$ for which
$0\le 3L\frac{\sigma}{\beta}-\|F^{\frac{\sigma}{\beta}} \circ F_{2\frac{\sigma}{\beta}}(\hat\varphi)-\varphi\|_\infty\le 0.24$.}
		\label{histo2}
	\end{figure}
	
\subsection{Experiments}
In order to check how good the upper bound stated in Theorem~\ref{main_result} is, we have made the following experiment.
In a first step, we have fixed $\ell=20$ and $\sigma=1.1$, and assumed as given the parameters $L>0$, $N\in\mathbb{N}$,
$\alpha>0$, $\beta>0$, $k\in\mathbb{N}$, $a_1,\ldots,a_k\in]0,\alpha[$, $b_1,\ldots,b_k>\beta$, $c_1,\ldots,c_k\in]0,\ell[$.

Firstly, we have used the parameters $L>0$, $N\in\mathbb{N}$ to generate a random $L$-Lipschitz function $\varphi$ in the following way.
We have randomly chosen and sorted in ascending order $N$ points $x_1,\ldots,x_N$ in the open interval $]0,\ell[$, with uniform distribution. Hence we have obtained the following decomposition: $]0,\ell[=]0,x_1]\cup]x_1,x_2]\cup\ldots\cup]x_N,\ell[$. We have defined our Lipschitz function $\varphi$ to be $0$ outside $]0,\ell[$.
After setting $(x_0,y_0)=(0,0)$ and $(x_{N+1},y_{N+1})=(\ell,0)$, for $i\in \{1,\ldots,N+1\}$
the value $y_i$ of the function at $x_i$ has been randomly chosen, with uniform distribution in an interval that allows for an $L$-Lipschitz extension to $[0,\ell]$ of the function, i.e., $\left[\max\{y_{i-1}-L(x_i-x_{i-1}),-L(\ell-x_i)\},\min\{y_{i-1}+L(x_i-x_{i-1}),L(\ell-x_i)\}\right]$.
Finally, the graph of the Lipschitz function $\varphi$ on $[0,\ell]$ has been obtained by connecting each point $(x_{i-1},y_{i-1})$ to $(x_i,y_i)$ with a segment, for $i\in \{1,\ldots,N+1\}$. We observe that $\varphi$ constructed this way is an $L$-Lipschitz function.

Secondly, we have
used the parameters $\alpha>0$, $\beta>0$, $k\in\mathbb{N}$, $a_1,\ldots,a_k\in]0,\alpha[$, $b_1,\ldots,b_k>\beta$, $c_1,\ldots,c_k\in]0,\ell[$
to generate a noise function as follows.
We have considered the mother function $\psi$ defined by setting $\psi(x):=e^{1-\frac{1}{1-x^2}}$ for $x\in\left]-1,1\right[$ and $\psi(x):=0$ for  $x\notin\left]-1,1\right[$.
For each $L$-Lipschitz function $\varphi$ produced in the previously described way, we have considered the function $\hat\varphi=\varphi+\sum\limits_{i=1}^k a_i\psi(b_i(x-c_i))$.

In Figures~\ref{fig_L_functions_k_3} and \ref{fig_L_functions_k_7} some examples of the functions we have produced are displayed, for $N=3$ and $N=7$, respectively.
In each figure, the functions are displayed without noise (left) and with added noise (right).
\begin{figure}[t]
	\begin{center}
		\begin{tabular}{c c}
			\includegraphics[width=4.8cm]{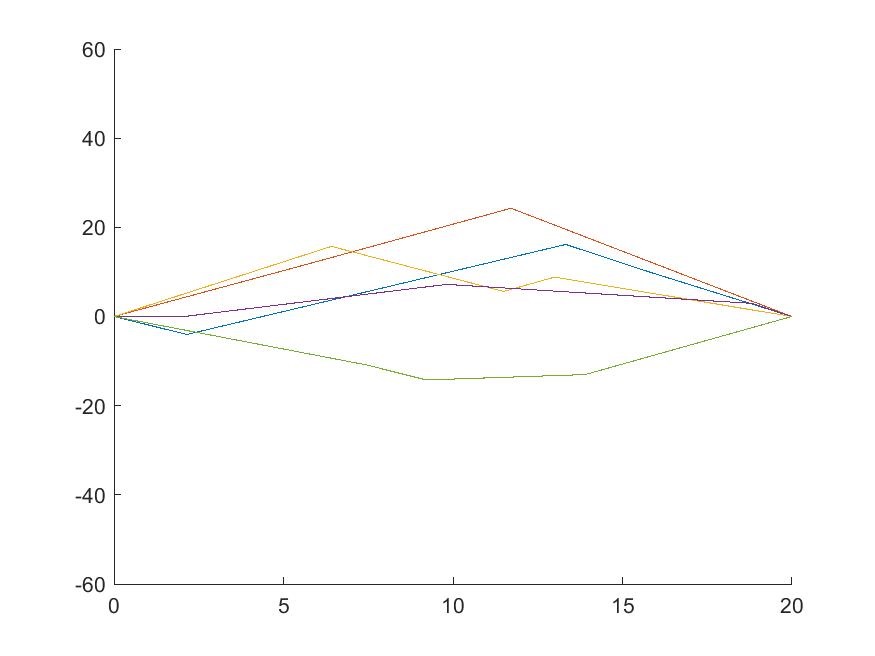} &
			\includegraphics[width=4.8cm]{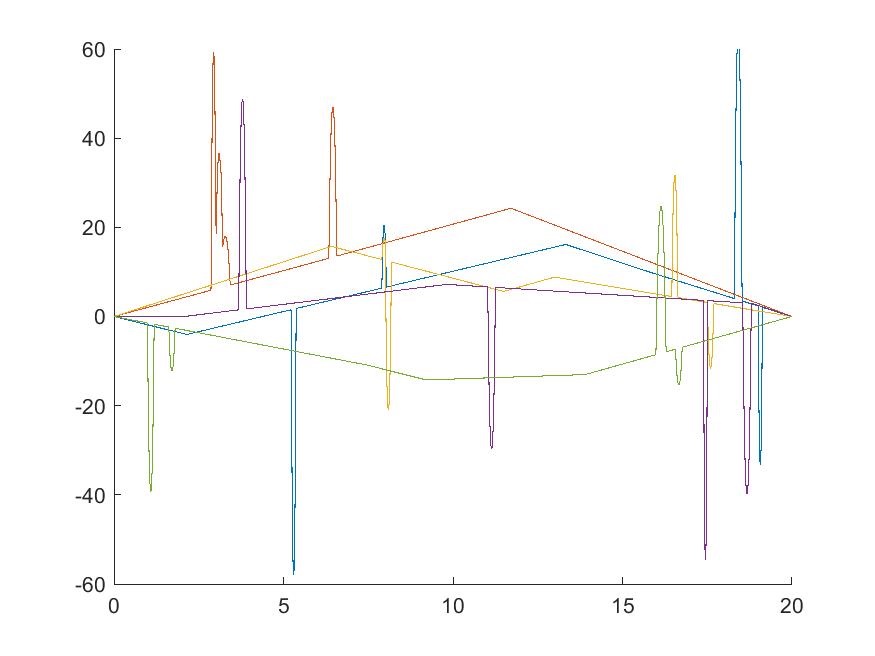}
			\end{tabular}
	\end{center}
	\caption{Five examples of the functions we have produced for $N=3$.
In each figure, the functions are displayed without noise (left) and with added noise (right).}
	\label{fig_L_functions_k_3}
	\end{figure}

\begin{figure}[t]
	\begin{center}
		\begin{tabular}{c c}
			\includegraphics[width=4.8cm]{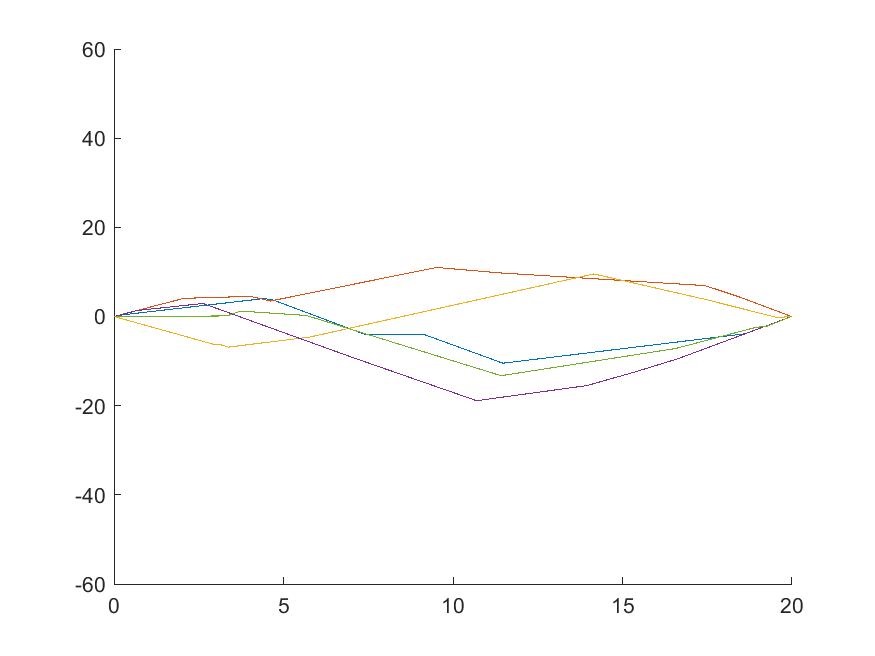} &
			\includegraphics[width=4.8cm]{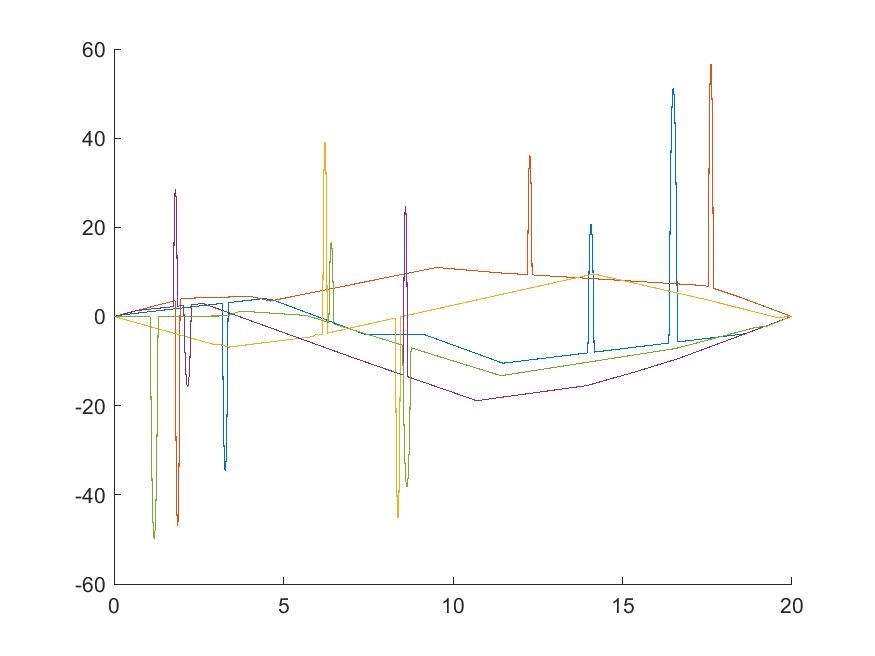}
			\end{tabular}
	\end{center}
	\caption{Five examples of the functions we have produced for $N=7$.
In each figure, the functions are displayed without noise (left) and with added noise (right).}
	\label{fig_L_functions_k_7}
	\end{figure}

In a second step, we have fixed $\ell=20$ and $\sigma=1.1$ once again, and considered a probabilistic model assuming that $\alpha$, $\beta$, $L$ are given and the values $N$, $k$, $a_i$, $b_i$, $c_i$ are random variables.
In this setting, we have compared the noise $\|\hat\varphi-\varphi\|_\infty$ with the probabilistic upper bound stated in Theorem~\ref{main_result} and the value $\left\|
F^{\frac{\sigma}{\beta}}\circ F_{2\frac{\sigma}{\beta}}(\hat\varphi)-\varphi
\right\|_\infty$, representing the reduced noise that we can obtain by applying our method.
In order to average our results, for each triplet $(\alpha,\beta,L)$ with $\alpha\in\{50, 55, 60, \ldots, 100\}$, $\beta\in\{3, 4, 5, \ldots, 13\}$, and $L\in\{1,2, \ldots, 10\}$, we have randomly generated $100$ examples of an $L$-Lipschitz function $\varphi$ and its noisy version $\hat\varphi$, by randomly choosing the parameters $N$, $k$, $a_i$, $b_i$, $c_i$ according to the following distributions:
\begin{itemize}
	\item $N \sim \mathrm{Unif}_{\{1,\ldots,10\}}$
	\item $k \sim \mathrm{Unif}_{\{1,\ldots,10\}}$
	\item $a_i \sim \mathrm{Unif}_{\left]0,\alpha\right[}$ for $i=1,\ldots,k$
	\item $b_i \sim \mathrm{Unif}_{\left]\beta,20\right[}$ for $i=1,\ldots,k$
	\item $c_i \sim \mathrm{Unif}_{\left]3\frac{\sigma}{\beta},\ell-3\frac{\sigma}{\beta}\right[}$ for $i=1,\ldots,k$.
\end{itemize}
Then for the chosen values of each one of the three variables $\alpha,\beta,L$, we have computed the mean, with respect to the other two variables,
of the average of the noise $\|\hat\varphi-\varphi\|_\infty$
and the average of the reduced noise $\left\|
F^{\frac{\sigma}{\beta}}\circ F_{2\frac{\sigma}{\beta}}(\hat\varphi)-\varphi
\right\|_\infty$ obtained by our method, both evaluated for $(\varphi,\hat\varphi)$ varying in the set of cardinality $100$ that we have produced. We have also computed the mean of the probabilistic upper bound stated in Theorem~\ref{main_result} with respect to the same two variables.
The results are displayed in Figures~\ref{fig_HISTOGRAM_VARALPHA_PROB}, \ref{fig_HISTOGRAM_VARBETA_PROB}, and \ref{fig_HISTOGRAM_VARL_PROB}.
We remind the reader that $\alpha$ and $\beta$ respectively express the maximum height and the thinness of the noise bumps, while $L$ is a Lipschitz constant for each function $\varphi$ we are interested in. These results illustrate the effectiveness of the use of GENEOs in the reduction of impulsive noise.
	
\begin{figure}[t]
	\begin{center}
		\includegraphics[width=7cm]{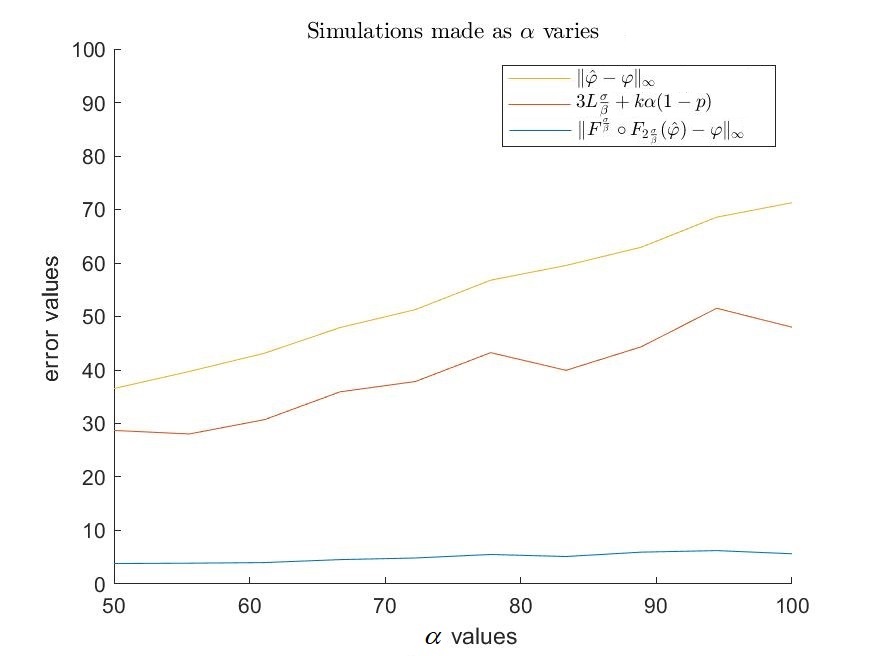}
	\end{center}
	\caption{Plots of the means of the averaged values of $\|\hat\varphi-\varphi\|_\infty$ (yellow), the means of
$3 L\frac{\sigma}{\beta}+k\alpha\left(1-\left(1-8\frac{(k-1)}{\ell}\frac{\sigma}{\beta}\right)^{k}\right)$ (brown), and the means of the averaged values of $\left\|
		F^{\frac{\sigma}{\beta}}\circ F_{2\frac{\sigma}{\beta}}(\hat\varphi)-\varphi\right\|_\infty$ (blue) for $\beta\in\{3, 4, 5, \ldots, 13\}$ and $L\in\{1,2, \ldots, 10\}$, when $\alpha$ varies in the set $\{50, 55, 60, \ldots, 100\}$.}
	\label{fig_HISTOGRAM_VARALPHA_PROB}
	\end{figure}

\begin{figure}[t]
	\begin{center}
		\includegraphics[width=7cm]{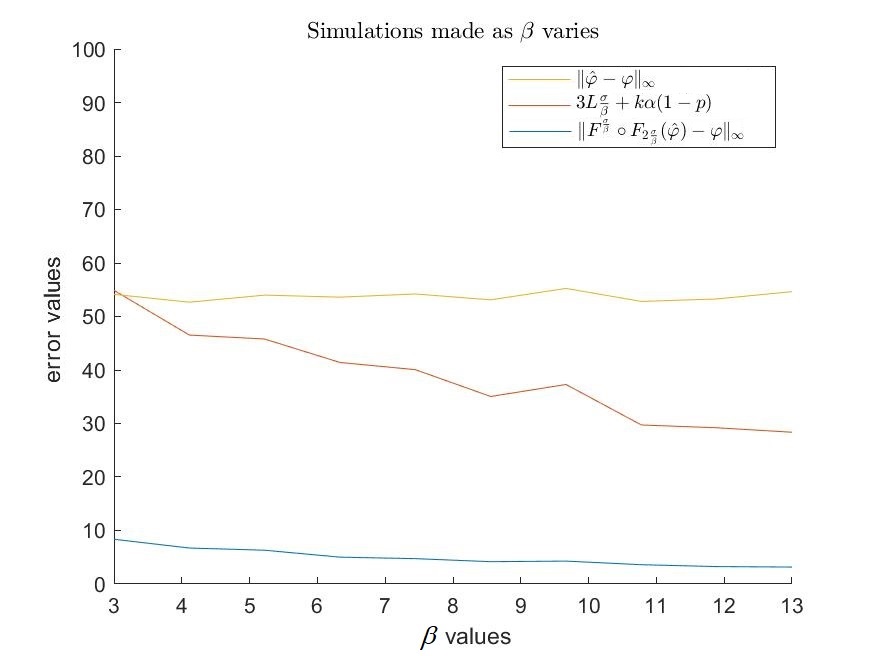}
	\end{center}
	\caption{Plots of the means of the averaged values of $\|\hat\varphi-\varphi\|_\infty$ (yellow), the means of
$3 L\frac{\sigma}{\beta}+k\alpha\left(1-\left(1-8\frac{(k-1)}{\ell}\frac{\sigma}{\beta}\right)^{k}\right)$ (brown), and the means of the averaged values of $\left\|
		F^{\frac{\sigma}{\beta}}\circ F_{2\frac{\sigma}{\beta}}(\hat\varphi)-\varphi\right\|_\infty$ (blue) for $\alpha\in\{50, 55, 60, \ldots, 100\}$ and $L\in\{1,2, \ldots, 10\}$, when $\beta$ varies in the set $\{3, 4, 5, \ldots, 13\}$.}
	\label{fig_HISTOGRAM_VARBETA_PROB}
	\end{figure}

\begin{figure}[t]
	\begin{center}
		\includegraphics[width=7cm]{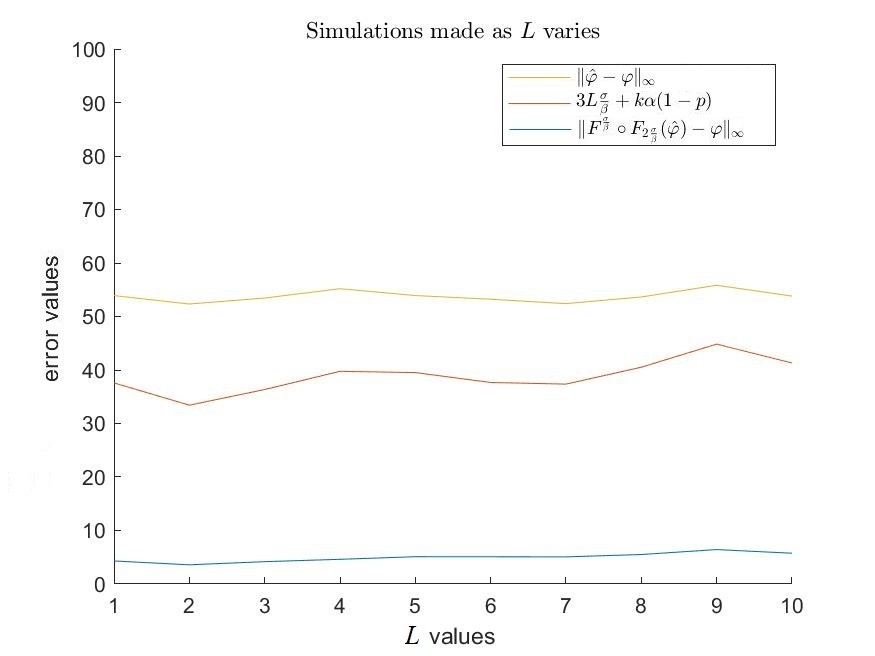}
	\end{center}
	\caption{Plots of the means of the averaged values of $\|\hat\varphi-\varphi\|_\infty$ (yellow), the means of
$3 L\frac{\sigma}{\beta}+k\alpha\left(1-\left(1-8\frac{(k-1)}{\ell}\frac{\sigma}{\beta}\right)^{k}\right)$ (brown), and the means of the averaged values of $\left\|
		F^{\frac{\sigma}{\beta}}\circ F_{2\frac{\sigma}{\beta}}(\hat\varphi)-\varphi\right\|_\infty$ (blue) for $\alpha\in\{50, 55, 60, \ldots, 100\}$ and $\beta\in\{3, 4, 5, \ldots, 13\}$, when $L$ varies in the set $\{1,2, \ldots, 10\}$.}
	\label{fig_HISTOGRAM_VARL_PROB}
	\end{figure}

\section{Conclusion}
In our paper we have proved a stability property for persistence diagrams of functions from $\mathbb{R}$ to $\mathbb{R}$, in the presence of impulsive noise.
This property shows that TDA can also be of use when noise drastically changes the topology of the sublevel sets of the filtering functions we are considering,
and stresses some new possible interaction between TDA and the theory of GENEOs. The experimental section shows that our approach is indeed able to remove the impulsive noise, in most of the cases. It would be interesting to check the possibility of extending our method to real-valued functions defined on $n$-dimensional domains, by selecting suitable GENEOs. We plan to devote our research to this topic in the future.

 \section*{Statements and Declarations}
This research has been partially supported by INdAM-GNSAGA.
The authors declare that they have no conflict of interest.
P.F. devised the project. All authors contributed to the manuscript. All authors read and approved the final manuscript.
The authors of this paper have been listed in alphabetical order.
Data available on request from the authors.
	

\begin{thebibliography}{10}
		
		\bibitem{AdAg19}
		Robert~J. Adler and Sarit Agami.
		\newblock Modelling persistence diagrams with planar point processes, and
		revealing topology with bagplots.
		\newblock {\em Journal of Applied and Computational Topology}, 3(3):139--183,
		Sep 2019.
		
		\bibitem{BeFrGiQu19}
		Mattia~G. Bergomi, Patrizio Frosini, Daniela Giorgi, and Nicola Quercioli.
		\newblock Towards a topological--geometrical theory of group equivariant
		non-expansive operators for data analysis and machine learning.
		\newblock {\em Nature Machine Intelligence}, 1(9):423--433, Sep 2019.
		
		\bibitem{BiDFFaal08}
		Silvia Biasotti, Leila De~Floriani, Bianca Falcidieno, Patrizio Frosini,
		Daniela Giorgi, Claudia Landi, Laura Papaleo, and Michela Spagnuolo.
		\newblock Describing shapes by geometrical-topological properties of real
		functions.
		\newblock {\em ACM Comput. Surv.}, 40(4):12:1--12:87, October 2008.
		
		\bibitem{BuChDeFaOuWa15}
		Micka{\"e}l Buchet, Fr{\'e}d{\'e}ric Chazal, Tamal~K. Dey, Fengtao Fan,
		Steve~Y. Oudot, and Yusu Wang.
		\newblock Topological analysis of scalar fields with outliers.
		\newblock In Lars Arge and J{\'a}nos Pach, editors, {\em 31st International
			Symposium on Computational Geometry (SoCG 2015)}, volume~34 of {\em Leibniz
			International Proceedings in Informatics (LIPIcs)}, pages 827--841, Dagstuhl,
		Germany, 2015. Schloss Dagstuhl--Leibniz-Zentrum f\"ur Informatik.
		
		\bibitem{Ca2009}
		Gunnar Carlsson.
		\newblock Topology and data.
		\newblock {\em Bull. Amer. Math. Soc. (N.S.)}, 46(2):255--308, 2009.
		
		\bibitem{CeDFFeal13}
		Andrea Cerri, Barbara Di~Fabio, Massimo Ferri, Patrizio Frosini, and Claudia
		Landi.
		\newblock Betti numbers in multidimensional persistent homology are stable
		functions.
		\newblock {\em Math. Methods Appl. Sci.}, 36(12):1543--1557, 2013.
		
		\bibitem{CeEtFr19}
		Andrea Cerri, Marc Ethier, and Patrizio Frosini.
		\newblock On the geometrical properties of the coherent matching distance in
		2{D} persistent homology.
		\newblock {\em Journal of Applied and Computational Topology}, 3(4):381--422,
		Dec 2019.
		
		\bibitem{CSEdHa07}
		David Cohen-Steiner, Herbert Edelsbrunner, and John Harer.
		\newblock Stability of persistence diagrams.
		\newblock {\em Discrete Comput. Geom.}, 37(1):103--120, 2007.
		
		\bibitem{CoEdHaMi10}
		David Cohen-Steiner, Herbert Edelsbrunner, John Harer, and Yuriy Mileyko.
		\newblock Lipschitz functions have {L}p-stable persistence.
		\newblock {\em Foundations of Computational Mathematics}, 10(2):127--139, Apr
		2010.

        \bibitem{CoFrQu22}
        Francesco Conti, Patrizio Frosini, and Nicola Quercioli.
        \newblock On the construction of Group Equivariant Non-Expansive Operators via permutants and symmetric functions,
        \newblock {\em Frontiers in Artificial Intelligence}, 5:1--11, 2022.
		
		\bibitem{EdHa08}
		Herbert Edelsbrunner and John Harer.
		\newblock Persistent homology---a survey.
		\newblock In {\em Surveys on discrete and computational geometry}, volume 453
		of {\em Contemp. Math.}, pages 257--282. Amer. Math. Soc., Providence, RI,
		2008.
		
		\bibitem{FaLeRiWaBaSi14}
		Brittany~Terese Fasy, Fabrizio Lecci, Alessandro Rinaldo, Larry Wasserman,
		Sivaraman Balakrishnan, and Aarti Singh.
		\newblock {Confidence sets for persistence diagrams}.
		\newblock {\em The Annals of Statistics}, 42(6):2301 -- 2339, 2014.
		
		\bibitem{FrJa16}
		Patrizio Frosini and Grzegorz Jab{\l}o{\'n}ski.
		\newblock Combining persistent homology and invariance groups for shape
		comparison.
		\newblock {\em Discrete Comput. Geom.}, 55(2):373--409, 2016.
		
		\bibitem{2001026}
		Mike~Earnest (https://math.stackexchange.com/users/177399/mike earnest).
		\newblock Average minimum distance between $n$ points generate i.i.d. with
		uniform dist.
		\newblock Mathematics Stack Exchange.
		\newblock URL:https://math.stackexchange.com/q/2001026 (version: 2016-11-09).
		
		\bibitem{Va08}
		Saeed~V. Vaseghi.
		\newblock {\em Impulsive Noise: Modelling, Detection and Removal}, chapter~13,
		pages 341--358.
		\newblock John Wiley and Sons, Ltd, 2008.
		
		\bibitem{ViFuKuSrBh20}
		Siddharth Vishwanath, Kenji Fukumizu, Satoshi Kuriki, and Bharath~K.
		Sriperumbudur.
		\newblock Robust persistence diagrams using reproducing kernels.
		\newblock In H.~Larochelle, M.~Ranzato, R.~Hadsell, M.~F. Balcan, and H.~Lin,
		editors, {\em Advances in Neural Information Processing Systems}, volume~33,
		pages 21900--21911. Curran Associates, Inc., 2020.
		
	\end{thebibliography}
\end{document}